\documentclass{article}
\usepackage{kr}
\usepackage{times}
\usepackage{xfrac}
\usepackage{soul}
\usepackage{url}
\usepackage[shortlabels]{enumitem}
\usepackage[hidelinks]{hyperref}
\usepackage[T1]{fontenc}
\usepackage[utf8]{inputenc}
\usepackage[small]{caption}
\usepackage{graphicx}
\usepackage{mathtools}
\usepackage{amsmath}
\usepackage{amssymb}
\usepackage{amsthm}
\usepackage{booktabs}
\usepackage{algorithm}
\usepackage{algorithmic}
\usepackage{amsfonts}
\usepackage{xcolor}
\usepackage{multicol}
\usepackage{multirow}
\usepackage{verbatim}
\usepackage{thm-restate}
\usepackage[normalem]{ulem}
\newcommand{\CPL}{\mathsf{CPL}}
\newcommand{\Luk}{{\normalfont{\textsf{\L}}}}

\newcommand{\RPL}{\mathsf{RPL}}

\newcommand{\conp}{\mathsf{coNP}}
\newcommand{\np}{\mathsf{NP}}
\newcommand{\DP}{\mathsf{DP}}

\newcommand{\Prop}{\mathsf{Pr}}

\newcommand{\LLuk}{\mathcal{L}_\Luk}
\newcommand{\LLukint}{\mathcal{L}^\mathbb{Q}_\Luk}
\newcommand{\LCPL}{\mathcal{L}_{\CPL}}

\newcommand{\Cmsf}{{\mathsf{C}}}

\newcommand{\Hmsf}{{\mathsf{H}}}

\newcommand{\Pmsf}{{\mathsf{P}}}

\newcommand{\Smsf}{{\mathsf{S}}}

\newcommand{\Xmsf}{{\mathsf{X}}}

\newcommand{\rmsf}{\mathsf{r}}

\newcommand{\cmbf}{\mathbf{c}}
\newcommand{\dmbf}{\mathbf{d}}

\newcommand{\imbf}{\mathbf{i}}
\newcommand{\vmbf}{\mathbf{v}}

\newcommand{\Nmbb}{\mathbb{N}}
\newcommand{\Pmbb}{\mathbb{P}}
\newcommand{\Qmbb}{\mathbb{Q}}
\newcommand{\Rmbb}{\mathbb{R}}

\newcommand{\Dmc}{\mathcal{D}}

\newcommand{\Hmc}{\mathcal{H}}

\newcommand{\Omc}{{\mathcal{O}}}

\newcommand{\Vmc}{{\mathcal{V}}}

\newcommand{\cl}{\mathsf{cl}}
\newcommand{\four}{\mathbf{4}}

\newcommand{\consvDashLuk}{\models^{\mathsf{cons}}_\Luk}
\newcommand{\consvDashCPL}{\models^{\mathsf{cons}}_\CPL}
\newcommand{\PropSol}{\mathcal{S}^\mathsf{p}}
\newcommand{\Sol}{\mathcal{S}}

\newcommand{\LukminSol}{\Sol^{\min}}
\newcommand{\ThminSol}{\Sol^\mathsf{Th}}
\newcommand{\interval}{\mathbf{int}}
\newcommand{\lineq}{\mathbb{LI}}

\newtheorem{theorem}{Theorem}
\newtheorem{proposition}{Proposition}

\newtheorem{example}{Example}
\newtheorem{definition}{Definition}
\newtheorem{convention}{Convention}
\theoremstyle{remark}
\newtheorem{remark}{Remark}

\title{Complexity of Abduction in Łukasiewicz Logic}

\author{%
Katsumi Inoue$^1$\and
Daniil Kozhemiachenko$^2$\\
\affiliations
$^1$National Institute of Informatics, Tokyo, Japan\\
$^2$Aix Marseille Univ, CNRS, LIS, Marseille, France
\emails
inoue@nii.ac.jp, daniil.kozhemiachenko@lis-lab.fr}

\begin{document}
\maketitle
\begin{abstract}
We explore the problem of explaining observations in contexts involving statements with truth degrees such as ‘the lift is loaded’, ‘the symptoms are severe’, etc. To formalise these contexts, we consider infinitely-valued Łukasiewicz fuzzy logic~$\Luk$. We define and motivate the notions of abduction problems and explanations in the language of~$\Luk$ expanded with ‘interval literals’ of the form $p\geq\cmbf$, $p\leq\cmbf$, and their negations that express the set of values a~variable can have. We analyse the complexity of standard abductive reasoning tasks (solution recognition, solution existence, and re\-le\-vance~/ necessity of hypotheses) in~$\Luk$ for the case of the full language and for the case of theories containing only disjunctive clauses and show that in contrast to classical propositional logic, the abduction in the clausal fragment has lower complexity than in the general case.% Finally, we discuss applications of Mixed-integer programming to abduction in Łukasiewicz logic thus enabling the use of existing reasoning procedures for $\Luk$ in the case of abductive reasoning.
\end{abstract}
\allowdisplaybreaks
\section{Introduction\label{sec:introduction}}
Abduction, deduction, and induction are three main forms of reasoning \cite{FlachKakas2000}. Abduction (finding explanations) has multiple applications in artificial intelligence, such as diagnosis \cite{ElAyebMarquisRusinowitch1993,JosephsonJosephson2009,Koitz-HristovWotawa2018}, commonsense reasoning~\cite{Paul1993,BhagavatulaLeBrasMalaviyaSakaguchiHoltzmanRashkinDowneyYihChoi2020}, formalisation of scientific reasoning~\cite{Magnani2001}, and machine learning~\cite{DaiXuYuZhou2019}. In \emph{logic-based abduction}~\cite{EiterGottlob1995}, the reasoning task is to find an explanation for an observation~$\chi$ from a~theory~$\Gamma$, i.e., a formula $\phi$ s.t.\ $\Gamma,\phi\models\chi$ but $\Gamma,\phi\not\models\bot$ (i.e., $\Gamma$ and $\phi$ should \emph{consistently entail} $\chi$).
%Usually, $\phi$ is somehow restricted because not every formula can explain a~given observation intuitively. Traditional restrictions are usually as follows (cf.\ discussion in~\cite[\S4.2]{Marquis2000HDRUMS} or~\cite[\S3.3]{Aliseda2006abductivereasoningbook}): $\phi$ should be syntactically restricted to be easily understandable, $\phi$~should not entail $\chi$ by itself nor contain atoms not occurring in $\Gamma\cup\{\phi\}$, and $\phi$ should constitute a weakest possible (or minimal) explanation. The first requirement is usually implemented by demanding explanations to be \emph{terms} (conjunctions of literals). In this case, the logically weakest solutions are the subset-minimal ones.

Observe, however, that in many applications, one needs not only to state whether a~formula is true but also to specify \emph{to which degree} it holds. E.g., a~lift may be safe to use only when loaded to at least $5\%$ of its maximal capacity (otherwise, its software will register it as empty) and to no more than $90\%$ of the capacity. Or the cruising speed of a~car can be defined as $60\ldots90~\mathrm{kmh}$ (with the maximal speed being $150~\mathrm{kmh}$). As classical logic has only two values, it is not well-suited to reason about contexts involving truth degrees. On the other hand, \emph{fuzzy logics} evaluate formulas in the real-valued interval $[0,1]$ and thus are much better suitable to formalising such contexts than classical logic. In particular, to formalise the examples above, we can set $v(l)\geq0.05$ and $v(s)\in[0.4,0.6]$ where the truth degrees of $l$ and $s$ denote, respectively, the load of the lift and the speed of the car w.r.t. its maximal speed.
\paragraph{Fuzzy Logic}
Originally~\cite{Zadeh1965,Zadeh1975} fuzzy logics were introduced to reason about imprecise statements such as ‘it is cold outside’, ‘the symptoms are severe’, etc. The values between $[0,1]$ are interpreted as degrees of truth from $0$ (absolutely false) to $1$ (absolutely true). Fuzzy logic has also been applied to reasoning about uncertainty (cf., e.g.,~\cite{HajekTulipani2001} and~\cite{BaldiCintulaNoguera2020}) and beliefs~\cite{RodriguezTuytEstevaGodo2022}. In knowledge representation and reasoning, fuzzy logics have found multiple applications in representing graded and fuzzy ontologies~\cite{Straccia2016}. In such ontologies, concept assertions and terminological axioms have degrees of truth. Moreover, fuzzy versions of description logics and their computational properties have been extensively investigated~\cite{BorgwardtPenaloza2012,Borgwardt2014PhD,BorgwardtDistelPenaloza2014DL,BorgwardtPenaloza2017}.

In addition to that, fuzzy logic has found multiple applications in artificial intelligence. Recent work on machine learning tries to combine perception by deep learning and symbolic knowledge representation. Neurosymbolic frameworks such as~\cite{DiligentiGoriSacca2017,BadreddinedAvila-GarcezSerafiniSpranger2022} adopt semantics of fuzzy logic to support learning and reasoning in real-world domains. \cite{vanKriekenAcarvanHarmelen2022} analyse how different fuzzy logic semantics affect the behaviour of learning.  Fuzzy logic has also been used in reasoning problems with knowledge graphs~\cite{ChenHuSun2022} and MaxSAT~\cite{HanikovaManyaVidal2023}. Furthermore, t-norms (functions used to interpret conjunctions in fuzzy logic) are applied for autonomous driving with requirements~\cite{StoianGiunchigliaLukasiewicz2023}.
% \paragraph{Abduction in Non-Classical Logics}
%DK: some contexts are difficult to formalise in the classical logic; mention paraconsistent abduction as an example.
\paragraph{Abduction in Fuzzy Logic}
Abduction in different systems of fuzzy logic has long attracted attention. To the best of our knowledge, it was first presented by~\cite{YamadaMukaidono1995}. There, the authors formalised abduction problems in the infinitely-valued Łukasiewicz fuzzy logic~$\Luk$ and proposed explanations in the form of \emph{fuzzy sets}, i.e., assignments of values from $[0,1]$ to propositional variables. This approach was further expanded by~\cite{Vojtas1999} to G\"{o}del and Product fuzzy logics. Solutions to abduction problems in multiple fuzzy logics were further systematised by~\cite{dAllonnesAkdagBouchon-Meunier2007} and~\cite{ChakrabortyKonarPalJain2013}. Abduction in fuzzy logic has found multiple applications, in particular, diagnosis~\cite{MiyataFuruhashiUchikawa1998}, machine learning~\cite{BergadanoCutelloGunetti2000}, fuzzy logic programming~\cite{Vojtas2001,Ebrahim2001}, decision-making and learning in the presence of incomplete information~\cite{MellouliBouchon-Meunier2003,Tsypyschev2017}, and robot perception~\cite{Shanahan2005}.
\paragraph{Contributions}
Still, the complexity of fuzzy abduction largely remains unexplored. To the best of our knowledge, the only discussion is given by~\cite{Vojtas1999}. There, the author explores fuzzy abduction in definite logic programmes (i.e., sets of rules of the form $\langle B\leftarrow A_1,\ldots,A_n,x\rangle$ where each $A_i$ is an atom and $x$~belongs to the real-valued interval $[0,1]$). In particular, the author claims (cf.~\cite[\S4.4]{Vojtas1999} for details) that ‘[a]s linear programming is lying in $\np$ complexity class (even much lower) as prolog does, to find minimal solutions for a definite [fuzzy logic programming abduction problem] \ldots does not increase the complexity and remains in~$\np$’. Thus, there seems to be no formal study of the complexity of abduction in fuzzy logic.
%Still, the complexity of fuzzy abduction largely remains unexplored. To the best of our knowledge, the only known result is that the complexity of abduction in definite logic programmes (i.e., rules of the form $B\leftarrow A_1,\ldots,A_n$ assigned with values where each $A_i$ is an atom) interpreted in Łukasiewicz, G\"{o}del, or Product logic is in~$\np$~\cite{Vojtas1999}. Thus, no systematic study of the computational properties of abductive reasoning has been done.

In this paper, we make a~step towards a~study of the complexity of abduction in fuzzy logic. We concentrate on the Łukasiewicz logic as it is one of the most expressive ones. In particular, it can express rational numbers and continuous linear functions over $[0,1]$ (cf.~\cite{McNaughton1951} for details). In addition, it still retains many classical relations between its conjunction, disjunction, implication, and negation. We formalise abduction in~$\Luk$ and explore its computational properties. Our contribution is twofold. First, we propose and motivate a~new form of solutions --- \emph{interval terms} --- that allow to express intervals of values a~variable is permitted to have. Second, we establish an (almost) complete characterisation of the complexity of standard reasoning problems (solution recognition, solution existence, relevance and necessity of hypotheses) for the case of theories containing arbitrary formulas and those comprised of disjunctive clauses.
\paragraph{Plan of the Paper}
The paper is structured as follows. In Section~\ref{sec:Luk}, we present the Łukasiewicz logic. In Section~\ref{sec:intervalterms}, we propose and motivate interval terms that we will use as solutions to Łukasiewicz abduction problems. Sections~\ref{sec:Lukabduction} and~\ref{sec:simpleclauseabduction} are dedicated to the study of the complexity of abductive reasoning in the Łukasiewicz logic and its clausal fragment. Finally, we summarise our results and provide a~plan for future work in Section~\ref{sec:conclusion}. Due to limited space, some proofs are put into the appendix of the full version~\cite{InoueKozhemiachenko2025arxiv}.% Omitted details of the proofs can be found in the appendix.
\section{Łukasiewicz Logic\label{sec:Luk}}
We begin with the language of Łukasiewicz logic ($\Luk$). We fix a~countable set $\Prop$ of propositional variables and define $\LLuk$ via the following grammar.
\begin{align*}
\LLuk\ni\phi&\coloneqq p\in\Prop\mid\neg\phi\mid(\phi\odot\phi)\mid(\phi\oplus\phi)\mid(\phi\!\rightarrow\!\phi)
\end{align*}
\begin{convention}[Notation]\label{conv:notation}
We use the following shorthands:
\begin{align*}
\top&\coloneqq p\oplus\neg p&\bot&\coloneqq p\odot\neg p&\phi\!\leftrightarrow\!\chi&\coloneqq(\phi\!\rightarrow\!\chi)\!\odot\!(\chi\!\rightarrow\!\phi)
\end{align*}

For a~set of formulas $\Gamma$ and a~formula~$\phi$, we write $\Prop(\phi)$ and $\Prop[\Gamma]$ to denote the set of all variables occurring in $\phi$ and $\Gamma$, respectively.

We use $\Rmbb$ and $\Qmbb$ to denote the sets of real and rational numbers, respectively. When dealing with intervals, square brackets mean that the endpoint is included in the interval, and round brackets that it is excluded. Lower index~$_\Qmbb$ means that the interval contains rational numbers only. E.g.,
\begin{align*}
[\sfrac{1}{2},\sfrac{2}{3}]&=\{x\mid x\in\Rmbb, x\geq\sfrac{1}{2},x\leq\sfrac{2}{3}\}\\
(\sfrac{1}{2},\sfrac{2}{3}]_\Qmbb&=\{x\mid x\in\Qmbb,x>\sfrac{1}{2},x\leq\sfrac{2}{3}\}
\end{align*}
\end{convention}
% In what follows, we will use $\LLuk$ to designate the $\triangle$-free fragment of $\LLuk$.

The semantics of $\Luk$ is given in the next definition.
\begin{definition}[Semantics of Łukasiewicz logic]\label{def:Luksemantics}
An~\emph{$\Luk$-va\-lu\-a\-tion} is a~function $v:\Prop\rightarrow[0,1]$ extended to the complex formulas as follows:
\begin{align*}
v(\neg\phi)&=1-v(\phi)\\
% v(\triangle\phi)&=\begin{cases}1&\text{if }v(\phi)=1\\0&\text{otherwise}\end{cases}\\
v(\phi\odot\chi)&=\max(0,v(\phi)+v(\chi)-1)\\
v(\phi\oplus\chi)&=\min(1,v(\phi)+v(\chi))\\
v(\phi\!\rightarrow\!\chi)&=\min(1,1-v(\phi)+v(\chi))
\end{align*}

We say that $\phi\in\LLuk$ is \emph{$\Luk$-valid ($\Luk\models\phi$)} if $v(\phi)=1$ for every $\Luk$-valuation $v$; $\phi$ is \emph{$\Luk$-satisfiable} if $v(\phi)=1$ for some $\Luk$-valuation $v$.

We define two notions of equivalence --- \emph{strong equivalence ($\phi\equiv_\Luk\chi$)} and \emph{weak equivalence ($\phi\simeq_\Luk\chi$)}:
\begin{align*}
\phi\equiv_\Luk\chi&\text{ iff }\forall v:v(\phi)=v(\chi)\\
\phi\simeq_\Luk\chi&\text{ iff }\forall v:v(\phi)=1\leftrightharpoons v(\chi)=1
\end{align*}

Given a~finite $\Gamma\!\subset\!\LLuk$, \emph{$\Gamma$ entails $\chi$ in~$\Luk$ ($\Gamma\models_\Luk\chi$)} iff $v(\chi)\!=\!1$ in every $v$ s.t.\ $v(\phi)\!=\!1$ for all $\phi\!\in\!\Gamma$, 
% \begin{align*}
% \forall\phi\in\Gamma~[\forall v[(v(\phi)=1)\Rightarrow v(\chi)=1]
% % \forall v[(\forall\phi\in\Gamma~v(\phi)=1)\Rightarrow v(\chi)=1]
% %KI: I prefer the expression as: 
% %KI: \forall\phi\in\Gamma~[\forall v[(v(\phi)=1)\Rightarrow v(\chi)=1]
% %DK: ok, corrected.
% \end{align*}
and that \emph{$\Gamma$ consistently entails $\chi$ in~$\Luk$ ($\Gamma\!\consvDashLuk\!\chi$)} iff $\Gamma\not\models_\Luk\bot$ and $\Gamma\models_\Luk\chi$.
% \begin{align*}
% \Gamma\not\models_\Luk\bot&&\text{and}&&\Gamma\models_\Luk\chi
% \end{align*}
\end{definition}

We note some important semantical properties of $\Luk$. First, every connective \emph{behaves classically on $\{0,1\}$}: in particular, $\oplus$ behaves like disjunction and $\odot$ like conjunction. Thus, we will call $\oplus$ \emph{strong disjunction} and $\odot$~\emph{strong conjunction}.

Second, \emph{deduction theorem does not hold for $\models_\Luk$}: indeed, while $p\models_\Luk p\odot p$, it is easy to show that $\Luk\not\models p\rightarrow(p\odot p)$ by setting $v(p)=\frac{1}{2}$. Third, $\odot$ and $\oplus$ \emph{are not idempotent}: setting $v(p)\!=\!\frac{1}{2}$, we have that $v(p\!\odot\!p)\!=\!0$ and $v(p\!\oplus\!p)\!=\!1$.

Still, $\neg$, $\odot$, $\oplus$, and $\rightarrow$ interact in an expected manner. In particular, it is easy to check that $v(\top)=1$ and $v(\bot)=0$ for every valuation and that the following pairs of formulas are indeed \emph{strongly equivalent}:
\begin{align}
\neg(\phi\odot\chi)&\equiv_\Luk\neg\phi\oplus\neg\chi&\neg(\phi\oplus\chi)&\equiv_\Luk\neg\phi\odot\neg\chi\nonumber\\
\neg(\phi\!\rightarrow\!\chi)&\equiv_\Luk\phi\odot\neg\chi&\neg\phi\oplus\chi&\equiv_\Luk\phi\rightarrow\chi\label{equ:demorganproperties}\\
\neg\neg\phi&\equiv_\Luk\phi&(\phi\odot\chi)\!\rightarrow\!\psi&\equiv_\Luk\phi\!\rightarrow\!(\chi\!\rightarrow\!\psi)\nonumber
\end{align}

It is also important to observe that \emph{weak} conjunction ($\wedge$) and disjunction ($\vee$) are definable as follows:
\begin{align}
\phi\vee\chi&\coloneqq(\phi\rightarrow\chi)\rightarrow\chi&\phi\wedge\chi&\coloneqq\neg(\neg\phi\vee\neg\chi)\label{equ:weakconnectives}
\end{align}
Using Definition~\ref{def:Luksemantics}, one can recover semantics of $\wedge$ and $\vee$:
\begin{align*}
v(\phi\!\wedge\!\chi)&=\min(v(\phi),v(\chi))&v(\phi\!\vee\!\chi)&=\max(v(\phi),v(\chi))
\end{align*}
In what follows, we will write $\Gamma,\phi\models_\Luk\chi$ as a~shorthand for $\Gamma\cup\{\phi\}\models_\Luk\chi$. Similarly, if the set of premises is given explicitly, we omit brackets. E.g.,~$\phi,\chi\models_\Luk\psi$ stands for $\{\phi,\chi\}\models_\Luk\psi$. Note, furthermore, that the comma in the set of premises can be equivalently interpreted as $\wedge$ and $\odot$ as $(\phi\wedge\chi)\simeq_\Luk(\phi\odot\chi)$. Hence, for $\Gamma=\{\phi_1,\ldots,\phi_n\}$, we have
\begin{align*}
\Gamma\models_\Luk\chi&\text{ iff }\bigodot^n_{i=1}\phi_i\models_\Luk\chi\text{ iff }\bigwedge^n_{i=1}\phi_i\models_\Luk\chi
\end{align*}

We finish the section by recalling the complexity of $\Luk$. It is known~\cite[Theorem~3.4]{Mundici1987} that satisfiability of arbitrary formulas in $\Luk$ is $\np$-complete while validity and entailment are $\conp$-complete~\cite[Corollary~4.1.3]{Hanikova2011MFL2} just as in classical propositional logic ($\CPL$).
\begin{proposition}\label{prop:Luknpcomplete}~
\begin{enumerate}[noitemsep,topsep=1pt]
\item $\Luk$-satisfiability is $\np$-complete.
\item Entailment in $\Luk$ is $\conp$-complete.
% \item Given a~finite $\Gamma\cup\{\chi\}\subseteq\LLuk$, it is $\conp$-complete to decide whether $\Gamma\models_\Luk\chi$.
\end{enumerate}
\end{proposition}
% Moreover, even though deduction theorem does not hold in general, its restricted form is valid:
% \begin{align}\label{equ:trianglededuction}
% \phi\models_\Luk\chi&\text{ iff }\Luk\models\triangle\phi\rightarrow\chi
% \end{align}
% \cite{Baaz1996}
\section{Interval Terms\label{sec:intervalterms}}
Before proceeding to the formal presentation of abduction, let us introduce the terms that we will be using in solutions. Traditionally, the form of solutions is restricted to conjunctions of literals (\emph{terms}). In this case, a~solution corresponds to a~statement of facts. Moreover, in this case, the logically weakest solutions are the subset-minimal ones. We begin with Łukasiewicz counterparts of classical terms and clauses.
\begin{definition}[Simple literals, clauses, and terms]\label{def:simpleclauses}~
\begin{itemize}[noitemsep,topsep=1pt]
\item A~\emph{simple literal} is a~propositonal variable or its negation.
\item A~\emph{simple clause} is a~\emph{strong disjunction} of simple literals, i.e., a~formula of the form $\bigoplus^n_{i=1}l_i$ for some $n\in\Nmbb$.
\item A~\emph{simple term} is a~\emph{strong conjunction} of simple literals, i.e., a~formula of the form $\bigodot^n_{i=1}l_i$ for some $n\in\Nmbb$.
\end{itemize}
\end{definition}

One can observe, however, that \emph{simple terms} from Definition~\ref{def:simpleclauses} are too restrictive if we want to use them as solutions for abduction problems. Indeed, a~simple term $\tau$ has value~$1$ iff all its literals have value $1$. But in a~context with fuzzy propositions, we might need to express statements such as ‘$p$ has value $\frac{2}{3}$’ or ‘$q$ has value at least $\frac{1}{4}$’. The following example illustrates this situation.
\begin{example}\label{example:thresholds}
Assume that we have a~lift with a weight sensor that controls two indicators: green and blue. The green indicator is on when the weight sensor detects the load of \emph{at least $\frac{1}{4}$} of its maximal capacity. Blue light is on when the lift is loaded to \emph{at most $\frac{2}{3}$} of its maximal capacity. We see that both indicators are lit.

Let us now formalise this problem in $\Luk$. We interpret the value of $c$ as the percentage of capacity to which the lift is loaded and translate the condition ‘the lift is loaded to at least $\sfrac{1}{4}$ of its capacity’ as $c\oplus c\oplus c\oplus c$. Observe that $v(c\oplus c\oplus c\oplus c)=1$ iff $v(c)\geq\sfrac{1}{4}$. To represent the other condition ‘the lift is loaded to at most~$\sfrac{2}{3}$ of its capacity’, we write $\neg c\oplus\neg c\oplus\neg c$. 
%KI: I could not understand what is "c" and why "c" or "\neg c" is used three times here.  Are there three sensors here?
%DK: The value of $c$ (from $0$ to $1$) stands for the load w.r.t.\ capacity. E.g., if $v(c)=0$, the lift is not loaded; if $v(c)=\frac{1}{2}$, then it is half-full; if $v(c)=1$, then it is fully loaded. We need to repeat $c$ (or $\neg c$) several times to be able to express the statements such as ‘the lift is loaded to \frac{1}{4} of its maximal capacity’.
We have that $v(\neg c\oplus\neg c\oplus\neg c)=1$ iff $v(\neg c)\geq\sfrac{1}{3}$, i.e., iff $v(c)\leq\sfrac{2}{3}$. Finally, we use $g$ and $b$ to represent that the green and blue lights are on. We obtain the following theory $\Gamma_\mathsf{lift}$ and observation $\chi_\mathsf{lift}$:
\begin{align*}
\Gamma_\mathsf{lift}&=\left\{\begin{matrix}(c\oplus c\oplus c\oplus c)\leftrightarrow g,(\neg c\oplus\neg c\oplus\neg c)\leftrightarrow b\end{matrix}\right\}
\\\chi_\mathsf{lift}&=g\odot b
\end{align*}

To explain why both indicators are on, we need to present a~formula $\phi$ s.t.\ $\Gamma_\mathsf{lift},\phi\consvDashLuk g\odot b$. For this, we need that $v(\neg c\oplus\neg c\oplus\neg c)=1$ and $v(c\oplus c\oplus c\oplus c)=1$ in every valuation that makes $\Gamma_\mathsf{lift}$ true. As we noted above, this requires that $v(c)\in[\frac{1}{4},\frac{2}{3}]$. 
% Indeed, since $v(p'\rightarrow(\neg p'\oplus\neg p'))=1$, we have that $v(p')\leq v(p')+v(p')$ and $v(c)\leq v(p')+v(p')$. Similarly, $v(\neg c\!\rightarrow\!(c\oplus c\oplus c))=1$ implies that $1-v(c)\leq v(c)+v(c)+v(c)$. 
On the other hand, a~simple term $\tau$ has value $1$ iff all its literals have value $1$ (i.e., all variables in $\tau$ should have value $0$ or $1$). Thus, there is no simple term $\tau$ containing $c$ or $\neg c$ s.t.\ $\Gamma_\mathsf{lift},\!\tau\!\consvDashLuk\!\chi$. Hence $\phi$ cannot be a~simple term.
\end{example}

One way to circumvent this problem is to adopt the proposal of~\cite{YamadaMukaidono1995} and~\cite{Vojtas1999} and define solutions to fuzzy abduction problems as sets of assignments of values to propositional variables. This approach, however, has a~drawback. In this setting, one can only express exact values of variables but not \emph{intervals} of their values. Now, observe from Example~\ref{example:thresholds} that any assignment of a~value from $[\frac{1}{4},\frac{2}{3}]$ to $c$ solves $\langle\Gamma_\mathsf{lift},\chi_\mathsf{lift}\rangle$. Thus, in the general case, it is impossible to generate the set of all solutions as there are infinitely many of them. %Of course, \emph{only one} value of $c$ actually takes place in the real world, so our guessed exact value of $c$ is likely to be false.
Moreover, it may be problematic to choose between different values.% To allow for comparison between different solutions, \cite{Vojtas1999} proposes to associate costs with hypotheses.
%DK: do we need the last sentence here?

In this section, we propose an alternative. For that, we define terms that allow us to express both exact values of variables and their intervals and compare different solutions w.r.t.\ entailment in~$\Luk$. %, thus avoiding the need for cost functions. 
Moreover, as we will see in Section~\ref{sec:Lukabduction}, every abduction problem will have only finitely many solutions in our setting.
\begin{definition}[Rational interval literals, terms, and clauses]\label{def:intervalterms}
Let $p\!\in\!\Prop$, $\lozenge\!\in\!\{\leq,\geq,<,>\}$, and $\cmbf\in[0,1]_\Qmbb$.
\begin{itemize}[noitemsep,topsep=1pt]
\item A \emph{rational interval literal} has the form $p\lozenge\cmbf$ or $\neg(p\lozenge\cmbf)$. The semantics of rational interval literals is as follows:
\begin{align*}
v(p\lozenge\cmbf)&=
\begin{cases}
1&\text{if }v(p)\lozenge\cmbf\\
0&\text{otherwise}
\end{cases}
&
v(\neg(p\lozenge\cmbf))&=1-v(p\lozenge\cmbf)
\end{align*}
For a~rational interval literal $p\lozenge\cmbf$, we call $\cmbf$ its \emph{boundary value} and call the set $\{v(p)\mid v(p\lozenge\cmbf)=1\}$ its \emph{permitted values}.
\item A \emph{rational interval term} has the form $\bigodot^n_{i=1}\lambda_i$ with each $\lambda_i$ being an interval literal.
\item A \emph{rational interval clause} has the form $\bigoplus^n_{i=1}\lambda_i$ with each $\lambda_i$ being an interval literal.
\end{itemize}
\end{definition}
\begin{convention}\label{conv:intervallanguage}
We use $\LLukint$ to denote the language obtained from $\LLuk$ by expanding it with rational interval literals. We will mostly write ‘interval literals (terms, clauses)’ instead of ‘rational interval literals (terms, clauses)’. The notions of $\Luk$-validity, satisfiability and entailment are preserved from Definition~\ref{def:Luksemantics}. We will also utilise notation from Convention~\ref{conv:notation} for $\LLukint$. Additionally, given an interval literal $p\lozenge\cmbf$, we use $p\blacklozenge\cmbf$ to denote $\neg(p\lozenge\cmbf)$. Finally, given a~literal $\lambda$ and a~term or clause $\varrho$, we write $\lambda\in\varrho$ to designate that $\lambda$ occurs in $\varrho$.
\end{convention}

We observe that it follows from~\eqref{equ:demorganproperties} that interval clauses and terms are dual in the following sense:
\begin{align}\label{equ:intervaltermsintervalclausesduality}
\neg\bigodot\limits^n_{i=1}(p_i\lozenge\cmbf_i)\!\equiv_\Luk\!\bigoplus\limits^n_{i=1}(p_i\blacklozenge\cmbf_i)
&&
\neg\bigoplus\limits^n_{i=1}(p_i\lozenge\cmbf_i)\!\equiv_\Luk\!\bigodot\limits^n_{i=1}(p_i\blacklozenge\cmbf_i)
\end{align}

\begin{remark}\label{rem:intervalterms}
From Definition~\ref{def:intervalterms}, it is clear that interval terms generalise simple terms in the following sense: for every simple term $\tau$, there is an interval term $\tau^\lozenge$ s.t.\ $\tau\simeq_\Luk\tau^\lozenge$. Indeed, given $\tau=\bigodot^m_{i=1}p_i\odot\bigodot^n_{j=1}\neg q_j$, we can define $\tau^\lozenge\!=\!\bigodot^m_{i=1}(p_i\!\geq\!\mathbf{1})\!\odot\!\bigodot^n_{j=1}(q_j\!\leq\!\mathbf{0})$.
% \begin{align*}
% \tau^\lozenge&\coloneqq\bigodot\limits^m_{i=1}(p_i\geq1)\odot\bigodot\limits^n_{j=1}(q_j\leq0)
% \end{align*}
Moreover,
\begin{align*}
p\leq\cmbf\models_\Luk p\leq\cmbf'&\text{ iff }\cmbf\leq\cmbf'&p\geq\cmbf\models_\Luk p\geq\cmbf'&\text{ iff }\cmbf\geq\cmbf'
\end{align*}
and similarly for the literals of the form $p<\cmbf$ and $p>\cmbf$.
\end{remark}

The idea of interval terms comes from a~logic first introduced by~\cite{Pavelka1979FL1,Pavelka1979FL2,Pavelka1979FL3} as an extension of $\Luk$ with constants for every \emph{real number}. It turns out, however, (cf.~\cite[\S3.3]{Hajek1998} for details) that if one adds constants only for \emph{rational} numbers, the resulting logic (‘Rational Pavelka logic’ or $\RPL$) will have the same expressivity as the original one. %Moreover, $\RPL$ can be polynomially reduced to~$\Luk$: namely, for every set of $\RPL$-formulas $\Gamma\cup\{\chi\}$, there is $\Gamma^\Luk\cup\{\chi^\Luk\}\subseteq\LLuk$ s.t.\ $\Gamma\models_\Luk\chi$ iff $\Gamma^\Luk\models_\Luk\chi^\Luk$~\cite[Lemmas~3.3.11 and~3.3.13]{Hajek1998}.
To simulate the two-valued behaviour of interval terms, one also needs to introduce the ‘Delta operator’ $\triangle$ proposed by~\cite{Baaz1996} with the following semantics: $v(\triangle\phi)=1$ if $v(\phi)=1$ and $v(\triangle\phi)=0$, otherwise. Now, using the following equivalences
\begin{align*}
p\leq\cmbf&\equiv_\Luk\triangle(p\rightarrow\cmbf)&p\geq\cmbf&\equiv_\Luk\triangle(\cmbf\rightarrow p)\nonumber\\
p<\cmbf&\equiv_\Luk\neg\triangle(\cmbf\rightarrow p)&p>\cmbf&\equiv_\Luk\neg\triangle(p\rightarrow\cmbf)
\end{align*}
and expanding the constraint tableaux calculus of~\cite{Haehnle1999} with rules for $\triangle$ and rational constants, we have that the satisfiability and validity of $\LLukint$-formulas have the same complexity as those of~$\Luk$.
\begin{proposition}\label{prop:Lukintnpcomplete}~
\begin{enumerate}[noitemsep,topsep=1pt]
\item $\Luk$-satisfiability of $\LLukint$-formulas is $\np$-complete.
\item Entailment in~$\Luk$ of $\LLukint$-formulas is $\conp$-complete.
\end{enumerate}
\end{proposition}

We note briefly that interval terms allow us to express not only the values of variables but also of arbitrary formulas  (cf.~\cite{Flaminio2007} for an alternative approach). For \mbox{$\phi\!\in\!\LLuk$} and $p\!\in\!\Prop$ s.t.\ $p\!\notin\!\Prop(\phi)$, $c\!\in\![0,1]_\Qmbb$, and $\lozenge\!\in\!\{\leq,<,\geq,>\}$, we have $v(\phi)\lozenge c$ iff \mbox{$v((\phi\leftrightarrow p)\odot p\lozenge\cmbf)=1$}.

We finish the section by establishing the complexity of the entailment of interval terms.
\begin{restatable}{proposition}{intervaltermcomplexity}\label{prop:intervaltermcomplexity}
Let $\Gamma\cup\{\chi\}\subseteq\LLukint$ be finite and $\tau$ and $\tau'$ be interval terms. Then the following statements hold.
\begin{enumerate}[noitemsep,topsep=1pt]
\item It takes polynomial time to decide whether $\tau\models_\Luk\tau'$.
\item It is $\conp$-complete to decide whether $\tau\models_\Luk\chi$.
\item It is $\conp$-complete to decide $\Gamma,\tau\models_\Luk\tau'$.
\end{enumerate}
\end{restatable}
\section{Abduction in Łukasiewicz Logic\label{sec:Lukabduction}}
Let us now present abduction in $\Luk$. Our idea is to use interval terms as solutions to problems $\Pmbb=\langle\Gamma,\chi,\Hmsf\rangle$. Here, $\Hmsf$ is the set of hypotheses (interval literals) that one can use to build solutions. One can restrict it in two ways. First, one may allow \emph{arbitrary} interval literals over a~given \emph{finite} set of variables. Second, one can explicitly define a~finite set of interval literals. We choose the second option, as the first one leads to \emph{infinite} sets of solutions (cf.~Example~\ref{example:thresholds}).
\begin{definition}[$\Luk$-abduction problems and solutions]\label{def:Lukabduction}~
\begin{itemize}[noitemsep,topsep=1pt]
\item An \emph{$\Luk$-abduction problem} is a~tuple $\Pmbb=\langle\Gamma,\chi,\Hmsf\rangle$ with $\Gamma\cup\{\chi\}$ a~finite set of $\LLukint$-formulas, %$\Gamma\not\models_\Luk\chi$, 
%KI: $\Gamma\not\models_\Luk\chi$ is not necessary in the abduction problem, if we do not know if the observation is really entailed or not (i.e., if we include deduction as a special case of abduction).
%DK: ok, let's remove it.
and $\Hmsf$ a~finite set of \emph{interval literals}. We call $\Gamma$ a~\emph{theory}, $\chi$~an \emph{observation}, and members of $\Hmsf$ \emph{hypotheses}.
\item An \emph{$\Luk$-solution} of $\Pmbb$ is an~\emph{interval term} $\tau$ composed of hypotheses s.t.\ $\Gamma,\tau\consvDashLuk\chi$.
\item A solution is \emph{proper} if $\tau\not\models_\Luk\chi$.
\item A~proper solution $\tau$ is \emph{$\models_\Luk$-minimal (entailment-minimal)} if there is no proper solution $\sigma$ s.t.\ $\tau\models_\Luk\sigma$ and $\sigma\not\simeq_\Luk\tau$.
\item A~proper solution $\tau$ is \emph{theory-minimal} if there is no proper solution $\sigma$ s.t.\ $\Gamma,\sigma\not\models_\Luk\tau$ and $\Gamma,\tau\models_\Luk\sigma$.
\end{itemize}
\end{definition}
\begin{convention}
Given an abduction problem $\Pmbb$, we will use $\Sol(\Pmbb)$, $\PropSol(\Pmbb)$, $\LukminSol(\Pmbb)$, and $\ThminSol(\Pmbb)$ to denote the sets of all solutions, all proper solutions, all $\models_\Luk$-minimal solutions, and all theory-minimal solutions of $\Pmbb$, respectively.
\end{convention}

In the definition above, it is evident that there are finitely many (at most exponentially many in the size of $\Hmsf$) solutions for each abduction problem. We also present two notions of minimal solutions. Entailment-minimality corresponds to \emph{subset-minimality} by~\cite{EiterGottlob1995} % (cf.~\cite{Aliseda2006abductivereasoningbook} for a~similar notion)
in the setting of Łukasiewicz logic.
% As we saw in Proposition~\ref{prop:intervaltermcomplexity}, entailment of (satisfiable) interval terms can be reduced to the containment of intervals of one term in the other.
Theory-minimal solutions correspond to \emph{least specific} solutions in the terminology of~\cite{Stickel1990,SakamaInoue1995} and \emph{least presumptive} solutions in the terminology of~\cite{Poole1989}. Theory-minimal solutions can also be seen as duals of \emph{theory prime implicates} by~\cite{Marquis1995}.

In addition, it is easy to see that even though a~theory-minimal solution is entailment-minimal, the converse is not always the case. Indeed, let $\Pmbb=\langle\{p\vee q,r\},q\wedge r\rangle$. One can see that there are two entailment-minimal solutions: $p\leq\mathbf{0}$ and~$q\geq\mathbf{1}$. Note, however, that $p\vee q,r,p\leq\mathbf{0}\models_\Luk q\geq\mathbf{1}$. Thus, $p\leq\mathbf{0}$ \emph{is not theory-minimal}.

Let us now see how we can solve abduction problems using interval terms. Recall Example~\ref{example:thresholds}.
\begin{example}\label{example:thresholdssolution}
%DK: I wonder what is better here: to assume that we can have an exact solution, or to consider a situation when we have a mismatch in proportions. E.g., when $\Hmsf$ is built only with quarters / thirds.
We continue Example~\ref{example:thresholds}. We need to formulate an abduction problem $\Pmbb_\mathsf{lift}=\langle\Gamma_\mathsf{lift},\chi_\mathsf{lift},\Hmsf_\mathsf{lift}\rangle$. $\Gamma_\mathsf{lift}$ and $\chi_\mathsf{lift}$ are already given in Example~\ref{example:thresholds}. It remains to form the set of hypotheses we are allowed to use. Assume for simplicity that we can measure the load of our lift in twelfths of its capacity. Thus, we can set
\begin{align*}
\Hmsf_\mathsf{lift}&=\{c\lozenge\mathbf{\tfrac{i}{12}}\mid\lozenge\in\{\leq,\geq,<,>\}\text{ and }i\in\{0,\ldots,12\}\}
\end{align*}
It is now easy to check that $(c\geq\mathbf{\tfrac{3}{12}})\odot(c\leq\mathbf{\tfrac{8}{12}})$ is indeed the theory-minimal solution of $\Pmbb_\mathsf{lift}$. Moreover, as expected,
\begin{align*}
\Sol(\Pmbb_\mathsf{lift})&=\left\{(c\!\triangleright\!\mathbf{\tfrac{i}{12}})\!\odot\!(c\!\triangleleft\!\mathbf{\tfrac{i'}{12}})\left|\begin{matrix}~\triangleleft\!\in\!\{\leq,<\},~\triangleright\!\in\!\{\geq,>\},\\[.3em]~\imbf\!\geq\!3,~\imbf'\!\leq\!8,~\imbf\!<\!\imbf'\end{matrix}\right.\!\right\}\cup\\[.4em]
&\quad~~\{(c\leq\mathbf{\tfrac{i}{12}})\odot(c\geq\mathbf{\tfrac{i}{12}})\mid3\leq\imbf\leq8\}
\end{align*}
That is, given any interval $\Dmc\subseteq\left[\frac{1}{4},\frac{2}{3}\right]$, every interval term~$\tau$ s.t.\ $v(\tau)=1$ iff $v(c)\in\Dmc$ is a~solution to $\Pmbb_\mathsf{lift}$.
\end{example}

It is also important to observe that the set of solutions depends not only on the variables in $\Hmsf$ but on their permitted values as well. Indeed, if in Example~\ref{example:thresholdssolution}, we supposed that we can measure the load in \emph{tenths} of the maximal capacity (i.e., if $\Hmsf=\{c\lozenge\mathbf{\tfrac{i}{10}}\mid\lozenge\in\{\leq,\geq,<,>\}\text{ and }i\in\{0,\ldots,10\}\}$), the theory-minimal solution would be $(c\!\geq\!\mathbf{\tfrac{3}{10}})\!\odot\!(c\!\leq\!\mathbf{\tfrac{6}{10}})$. Hence, to test a~solution $\tau$ for minimality, we need to calculate the nearest permitted value from the boundary value for every interval literal in~$\tau$.
%KI: The "permitted values" should be explained here since there is no definition.  Moreover, we need to calculate the nearest permitted value (rational number) from the boundary value, which should be mentioned too.
%DK: I added the notion of boundary values and permitted values to Definition 3 and your remark to the text.

Let us now establish the complexity of abductive reasoning. We will consider three standard tasks:
\begin{itemize}[noitemsep,topsep=1pt]
\item \emph{solution recognition} --- given a~problem $\Pmbb$ and an interval term~$\tau$, determine whether $\tau$ is a~(proper, entailment-minimal, or theory-minimal) solution;
\item \emph{solution existence} --- given a~problem $\Pmbb$, determine whether $\Sol(\Pmbb)=\varnothing$;
\item \emph{relevance and necessity of hypotheses} --- given a~problem $\Pmbb$ and a~hypothesis $\lambda$, determine whether there is a~solution where it occurs and whether it occurs in all solutions.
\end{itemize}

In our proofs, we will use reductions from the \emph{classical} abductive reasoning. In the following definition, we recall the notion of classical abduction problem and solutions. We adapt the definitions from~\cite{EiterGottlob1995} and~\cite{CreignouZanuttini2006} for our notation.  We use terminology and notation for $\CPL$ analogous to ones introduced for~$\Luk$, e.g.,\ speaking about $\CPL$-validity and using $\models_\CPL$ for the classical entailment relation.
\begin{definition}\label{def:CPLabductiveproblem}
Let $\LCPL$ be the propositional language over $\{\neg,\wedge,\vee\}$. A~\emph{classical abduction problem} is a~tuple $\Pmbb=\langle\Gamma,\chi,\Hmsf\rangle$ s.t.\ $\Gamma\!\cup\!\{\chi\}\!\subseteq\!\LCPL$ and $\Hmsf$ is a~set of simple literals.
\begin{itemize}[noitemsep,topsep=1pt]
\item A~\emph{solution} of $\Pmbb$ is a~weak conjunction $\tau$ of literals from $\Hmsf$ such that 
$\Gamma,\tau\consvDashCPL\psi$.
\item A~solution $\tau$ is \emph{proper} if $\tau\not\models_\CPL\psi$.
\item A~proper solution $\tau$ is \emph{$\models_\CPL$-minimal} if there is no proper solution $\phi$ s.t.\ $\tau\!\models_\CPL\!\phi$ and $\phi\not\models_\CPL\tau$.
\item A~proper solution $\tau$ is \emph{theory-minimal} if there is no proper solution $\phi$ s.t.\ $\Gamma,\tau\models_\CPL\phi$ and $\Gamma,\phi\not\models_\CPL\tau$.
\end{itemize}
\end{definition}

We will also need the following technical statement.
\begin{definition}\label{def:multiplicativetranslation}
Let $\phi\in\LCPL$. We define $\phi^\Luk$ as follows:
\begin{align*}
p^\Luk&=p&(\neg\phi)^\Luk&=\neg\phi^\Luk\\
(\phi\wedge\chi)^\Luk&=\phi^\Luk\odot\chi^\Luk&(\phi\vee\chi)^\Luk&=\phi^\Luk\oplus\chi^\Luk
\end{align*}
Given $\Gamma\subseteq\LCPL$, we set $\Gamma^\Luk=\{\phi^\Luk\mid\phi\in\Gamma\}$.
\end{definition}
\begin{restatable}{proposition}{CPLtoLuk}\label{prop:CPLtoLuk}
Let $\Gamma\cup\{\chi\}\subseteq\LCPL$. Then
\begin{align*}
\Gamma\models_\CPL\chi&\text{ iff }\Gamma^\Luk,\{p\vee\neg p\mid p\in\Prop[\Gamma\cup\{\chi\}]\}\models_\Luk\chi^\Luk
\end{align*}
\end{restatable}

The complexity results from this section are in Table~\ref{tab:complexitysolutions}.
\begin{table}
\centering
\begin{tabular}{lcc}
Recognition and existence & $\Luk$& $\CPL$ \\ \midrule %[.2em]
$\tau\in\Sol(\Pmbb)$? / $\tau\in\PropSol(\Pmbb)$? / $\tau\in\LukminSol(\Pmbb)$? & $\DP$& $\DP$ \\[.2em]
%$\tau\in\PropSol(\Pmbb)$? & $\conp$ & $\DP$ \\[.2em]
$\tau\in\ThminSol(\Pmbb)$?& in $\Pi^\Pmsf_2$& in $\Pi^\Pmsf_2$\\[.2em]
$\Sol(\Pmbb)\neq\varnothing$? / $\PropSol(\Pmbb)\neq\varnothing$?& $\Sigma^\Pmsf_2$& $\Sigma^\Pmsf_2$\\[.75em]
% \end{tabular}
% \vspace{.5em}
% %$\PropSol(\Pmbb)\neq\varnothing$? & $\Sigma^\Pmsf_2$ & $\Sigma^\Pmsf_2$ \\ [.2em]
% \begin{tabular}{lccc}
Relevance& $\Luk$& $\CPL$ \\\midrule
w.r.t.\ $\Sol(\Pmbb)$, $\PropSol(\Pmbb)$, $\LukminSol(\Pmbb)$& $\Sigma^\Pmsf_2$& $\Sigma^\Pmsf_2$\\[.2em]
w.r.t.\ $\ThminSol(\Pmbb)$ & in $\Sigma^\Pmsf_3$& in $\Sigma^\Pmsf_3$\\[.75em]
Necessity& $\Luk$ & $\CPL$ \\\midrule
w.r.t.\ $\Sol(\Pmbb)$, $\PropSol(\Pmbb)$, $\LukminSol(\Pmbb)$ & $\Pi^\Pmsf_2$& $\Pi^\Pmsf_2$\\[.2em]
w.r.t.\ $\ThminSol(\Pmbb)$ & in $\Pi^\Pmsf_3$& in $\Pi^\Pmsf_3$
\end{tabular}
\caption{Complexity of abductive reasoning problems. Unless specified otherwise, all results are completeness results.}
\label{tab:complexitysolutions}
\end{table}
\subsection{Solution Recognition\label{ssec:solutionrecognition}}
We begin with solution recognition. First, we show that recognition of arbitrary, proper, and entailment-minimal solutions is $\DP$-complete.
\begin{theorem}\label{theorem:arbitrarysolutionrecognitioncomplexity}
Let $\Pmbb=\langle\Gamma,\chi,\Hmsf\rangle$ be an $\Luk$-abduction problem and $\tau$ an interval term. Then, it is $\DP$-complete to decide whether $\tau\in\Sol(\Pmbb)$.
\end{theorem}
\begin{proof}
%DK: put into the appendix?
The membership follows immediately from Proposition~\ref{prop:Lukintnpcomplete} and the fact that $\tau\in\Sol(\Pmbb)$ iff $\Gamma,\tau\consvDashLuk\chi$. For the hardness, we provide a~reduction from the classical solution recognition which is $\DP$-complete~\cite[\S4]{EiterGottlob1995}. Let now $\Pmbb_\cl=\langle\Delta,\psi,\Hmsf^\cl\rangle$ be a~\emph{classical} abduction problem. Define $\Pmbb^\Luk=\langle\Delta^\sharp,\psi^\Luk,\Hmsf^\sharp\rangle$ as follows:
\begin{align*}
\Delta^\sharp&=\Delta^\Luk\cup\{p\vee\neg p\mid p\in\Prop[\Delta\cup\{\psi\}]\}\\
\Hmsf^\sharp&=\{p\geq\mathbf{1}\mid p\in\Hmsf^\cl\}\cup\{q\leq\mathbf{0}\mid q\in\Hmsf^\cl\}
\end{align*}
Furthermore, for $\tau=\bigwedge^m_{i=1}p_i\wedge\bigwedge^n_{j=1}\neg q_j$, 
% \begin{align*}
% \tau&=\bigwedge\limits^m_{i=1}p_i\wedge\bigwedge\limits^n_{j=1}\neg q_j
% \end{align*}
we define $\tau^\odot=\bigodot^m_{i=1}(p_i\geq\mathbf{1})\odot\bigodot^n_{j=1}(q_j\leq\mathbf{0})$ 
% \begin{align*}
% \tau^\odot&=\bigodot\limits^m_{i=1}(p_i\geq\mathbf{1})\odot\bigodot^n_{j=1}(q_j\leq\mathbf{0})
% \end{align*}
and show that $\tau\in\Sol(\Pmbb_\cl)$ iff $\tau^\odot\in\Sol(\Pmbb^\Luk)$.

Assume that $\tau\in\Sol(\Pmbb_\cl)$, i.e., $\Delta,\tau\consvDashCPL\psi$ and let $v$ be a~\emph{classical} valuation s.t.\ $v(\phi)=1$ for every $\phi\in\Delta$ and $v(\tau)=1$. Since $\odot$ and $\oplus$ behave on $\{0,1\}$ the same as $\wedge$ and $\vee$, it is clear that $v(\chi)=1$ for every $\chi\in\Delta^\sharp$ and $v(\tau^\odot)=1$ as well. Now assume further for the sake of contradiction that there is some $\Luk$-valuation $v^\Luk$ s.t.\ \mbox{$v^\Luk(\chi)=1$} for every $\chi\in\Delta^\sharp$, $v^\Luk(\tau^\odot)=1$, but $v(\psi^\Luk)\neq1$. But $v^\Luk$ must be \emph{classical}, i.e., assign only values from $\{0,1\}$ to all $p\!\in\!\Prop[\Delta\!\cup\!\{\psi\}]$ because $p\!\vee\!\neg p\!\in\!\Delta^\sharp$ and $v^\Luk(p\!\vee\!\neg p)\!=\!1$ iff $v(p)\in\{0,1\}$. This would mean that $v^\Luk$ witnesses $\Delta,\tau\not\models_\CPL\psi$, contrary to the assumption.

For the converse direction, given an interval term $\tau$, define $\tau_\cl\!=\!\bigwedge_{p\!\geq\!\mathbf{1}\in\tau}\!p\wedge\!\bigwedge_{q\!\leq\!\mathbf{0}\in\tau}\!\neg q$. One can check that $\tau\!\in\!\Sol(\Pmbb^\Luk)$ iff $\tau_\cl\!\in\!\Sol(\Pmbb)$.
\end{proof}
\begin{theorem}\label{theorem:propersolutionrecognitioncomplexity}
Let $\Pmbb=\langle\Gamma,\chi,\Hmsf\rangle$ be an $\Luk$-abduction problem and $\tau$ an interval term. Then, it is $\DP$-complete to decide whether $\tau\in\PropSol(\Pmbb)$.
\end{theorem}
\begin{proof}
First, we obtain the hardness via a~reduction from the arbitrary solution recognition in Łukasiewicz logic. Namely, let $\Pmbb=\langle\Gamma,\chi,\Hmsf\rangle$ be an $\Luk$-abduction problem. We show that $\tau\in\Sol(\Pmbb)$ iff $\tau$ is a~\emph{proper solution} of $\Pmbb_p=\langle\Gamma\cup\{p\},\chi\odot p,\Hmsf\rangle$ with $p\notin\Prop[\Gamma\cup\{\chi\}]$. Assume that $\tau\in\Sol(\Pmbb)$. As $p\notin\Prop(\chi)$ and $\tau$ is $\Luk$-satisfiable, it is clear that $\tau\not\models_\Luk\chi\odot p$. It is also clear that $\Gamma,p,\tau\models_\Luk\bot$ iff $\Gamma,\tau\models_\Luk\bot$ and $\Gamma,p,\tau\models_\Luk\chi\odot p$ iff $\Gamma,\tau\models_\Luk\chi$. Conversely, let $\tau\in\PropSol(\Pmbb_p)$. As $\Gamma,p,\tau\consvDashLuk\chi$, it is clear that $\Gamma,\tau\consvDashLuk\chi$.

For the membership, given $\tau$ and $\Pmbb$, we (i) use an $\np$ oracle to guess two $\Luk$-valuations $v_\mathsf{sat}$ and $v_\mathsf{prp}$ and check that they witness $\Gamma,\tau\not\models_\Luk\bot$ and $\tau\not\models_\Luk\chi$. At the same time, we (ii) conduct a~$\conp$ check that $\Gamma,\tau\models_\Luk\chi$. Note that this is possible because we do not need the result of (i) to do (ii). It follows that $\tau\in\PropSol(\Pmbb)$ iff both checks succeed.
\end{proof}

\begin{restatable}{theorem}{Lukminimalsolutionrecognitioncomplexity}\label{theorem:Lukminimalsolutionrecognitioncomplexity}
Let $\Pmbb=\langle\Gamma,\chi,\Hmsf\rangle$ be an $\Luk$-abduction problem and $\tau$ an interval term. Then, it is $\DP$-complete to decide whether $\tau\in\LukminSol(\Pmbb)$.
\end{restatable}
\begin{proof}[Proof sketch]
We begin with the hardness. We will provide a~reduction from the prime implicant recognition in classical logic which is $\DP$-complete~\cite[Proposition~111]{Marquis2000HDRUMS}. Let w.l.o.g.\ $\chi$ be classically satisfiable. We can prove that $\tau$ is a~prime implicant of $\chi$ iff $\tau^\odot$ is an $\models_\Luk$-minimal solution to $\Pmbb\!=\!\left\langle\Gamma,\chi^\Luk\!\odot\!q,\Hmsf\right\rangle$ with $q\!\notin\!\Prop(\chi)$ and
\begin{align*}
\Gamma&=\left\{p\vee\neg p\mid p\in\Prop(\chi\wedge q)\right\}\cup\{q\}\\
\Hmsf&=\{p\geq\mathbf{1}\mid p\in\Prop(\chi)\}\cup\{p\leq\mathbf{0}\mid p\in\Prop(\chi)\}
\end{align*}
To show this, we use that $v(p\!\vee\!\neg p)\!=\!1$ iff $v(p)\!\in\!\{0,1\}$, $v(\tau)\!\in\!\{0,1\}$ for every interval term~$\tau$ and $\Luk$-valuation~$v$, and $\neg$, $\odot$, and $\oplus$ behave classically on $\{0,1\}$. Conversely, given $\tau\in\LukminSol(\Pmbb)$, we can define a~prime implicant $\tau_\cl$ by replacing $\odot$ with $\wedge$, $r\geq\mathbf{1}$ with $r$, and $s\leq\mathbf{0}$ with $\neg s$. The reasoning is similar.

The proof of the membership utilises the fact that given an interval term $\tau=\lambda\odot\tau'$ with $\lambda=p\lozenge\cmbf$ and a~set of hypotheses $\Hmsf$, there are $\Omc(|\Hmsf|)$ interval terms $\sigma$ for which we need to check that $\tau\models_\Luk\sigma$ and $\sigma\not\models_\Luk\tau$ hold. Namely, we can either replace $\lambda$ with the ‘next weakest’ literal $\lambda^\flat_\Hmsf$ or (if $\lambda$ is itself the weakest in~$\Hmsf$) remove $\lambda$ altogether. For example, if $\Hmsf=\{p\lozenge\sfrac{\imbf}{\four}\mid\lozenge\in\{\leq,<,\geq,>\},\imbf\in\{\mathbf{0},\ldots,\mathbf{4}\}\}$ and $\lambda=p\leq\mathbf{\sfrac{1}{4}}$, then $\lambda^\flat_\Hmsf=p\leq\mathbf{\sfrac{2}{4}}$. Let us now use $\tau^\flat_\lambda$ to denote the term obtained from $\tau$ by replacing $\lambda$ with $\lambda^\flat_\Hmsf$ or removing $\lambda$ when $\lambda^\flat_\Hmsf\notin\Hmsf$. As $\tau\models_\Luk\tau^\flat_\lambda$, $\Gamma,\tau\not\models_\Luk\bot$ entails $\Gamma,\tau^\flat_\lambda\not\models_\Luk\bot$ and $\tau\not\models_\Luk\chi$ entails $\tau^\flat_\lambda\not\models_\Luk\chi$.

Now, we use an $\np$-oracle that guesses linearly many w.r.t.\ $|\Hmsf|$ $\Luk$-valuations $v_\mathsf{sat}$, $v_\mathsf{prp}$, and $v_\lambda$ (for each $\lambda\in\tau$) and verifies whether they witness (i) $\Gamma,\tau\not\models_\Luk\bot$, (ii) $\tau\not\models_\Luk\chi$, and (iii) $\Gamma,\tau^\flat_\lambda\not\models_\Luk\chi$. Parallel to that (as we do not need the results of (i)--(iii)), we conduct a~$\conp$ check that (iv) $\Gamma,\tau\models_\Luk\chi$. It follows from the definition of $\models_\Luk$-minimal solutions that $\tau$ is a~$\models_\Luk$-minimal solution iff the $\np$ and $\conp$ checks succeed.
\end{proof}

In the case of theory-minimal solutions, we establish membership in~$\Pi^\Pmsf_2$. We expect that this case is indeed harder than entailment-minimality, intuitively because the presence of the theory means we cannot readily identify a~polynomial number of candidates for better solutions to check. We leave the search for a~matching lower bound for future work and remark that, to the best of our knowledge, the complexity of the analogous problem in $\CPL$ is also unknown. 
\begin{restatable}{theorem}{theoryminimalrecognition}\label{theorem:theoryminimalrecognition}
It is in $\Pi^\Pmsf_2$ to decide, given an $\Luk$-abduction problem $\Pmbb$ and an interval term $\tau$, whether $\tau$ is a theory-minimal solution of $\Pmbb$.
\end{restatable}
\subsection{Solution Existence\label{ssec:solutionexistence}}
We now turn to establishing the complexity of the solution existence in $\Luk$-abduction problems. Note that a~problem may have solutions but no proper solutions. On the other hand, if a~problem has proper solutions, it will have entailment- and theory-minimal solutions as well. Thus, we will consider the complexity of arbitrary and proper solution existence.
\begin{theorem}\label{theorem:arbitrarysolutionexistence}
It is $\Sigma^\Pmsf_2$-complete to decide whether an $\Luk$-ab\-duc\-tion problem has a~(proper) solution.
\end{theorem}
\begin{proof}
Membership follows immediately from Theorem~\ref{theorem:arbitrarysolutionrecognitioncomplexity}. For hardness, we provide a~reduction from the solution existence for classical abduction problems $\Pmbb_\cl=\langle\Gamma_\cl,\chi_\cl,\Hmsf\rangle$ of the following form (below $\phi\in\LCPL$):\footnote{We write $\neg r\Leftrightarrow r'$ as a~shorthand for $(r\wedge\neg r')\vee(\neg r\wedge r')$.}
\begin{align}
\Gamma_\cl&=\{\!\neg\phi\!\vee\!(p\!\wedge\!\tau),\!\neg p\!\vee\!\tau\}\!\cup\!\{\neg r\!\Leftrightarrow\!r'\!\mid\! r\!\in\!\Prop(\phi)\!\setminus\!\Prop(p\!\wedge\!\tau)\!\}\nonumber\tag{$p\notin\Prop(\phi\wedge\tau)$, $\tau$ is a~weak conjunction of literals}\\
\chi_\cl&=p\wedge\tau\nonumber\\
\Hmsf&=\{r\mid r\!\in\!\Prop(\phi)\!\setminus\!\Prop(p\!\wedge\!\tau)\}\!\cup\!\{r'\mid\neg r\!\Leftrightarrow\!r'\!\in\!\Gamma_\cl\}
\label{equ:solutionexistenceCPL}
\end{align}
By~\cite[Theorem~4.2]{EiterGottlob1995}, determining the existence of classical solutions for these problems is $\Sigma^\Pmsf_2$-hard. We reduce $\Pmbb_\cl$ to $\Pmbb^\Luk=\langle\Gamma^\sharp,\chi^\Luk_\cl,\Hmsf^\sharp\rangle$ as was done in Theorem~\ref{theorem:arbitrarysolutionrecognitioncomplexity} (recall Definition~\ref{def:multiplicativetranslation} for $\Gamma^\Luk$ and $\chi^\Luk$) with
\begin{align}\label{equ:solutionexistenceLuk}
\Gamma^\sharp&=\Gamma^\Luk\!\cup\!\{p\!\vee\!\neg p\mid p\!\in\!\Prop[\Gamma_\cl\!\cup\!\{\chi_\cl\}]\}\nonumber\\
\Hmsf^\sharp&=\{s\geq\mathbf{1}\mid s\in\Hmsf\}\quad\chi^\Luk_\cl=p\odot\bigodot\limits_{l\in\tau} l
\end{align}
First let $\sigma$ be a~solution of $\Pmbb_\cl$. It is immediate from~\eqref{equ:solutionexistenceCPL} that $\sigma$ is a~\emph{proper} solution because we cannot use variables occurring in $\chi_\cl$. It now follows from Proposition~\ref{prop:CPLtoLuk} that $\Gamma^\sharp,\sigma^\Luk\consvDashLuk\chi^\Luk_\cl$. Furthermore, as $\sigma\not\models_\CPL\chi_\cl$, it is clear that $\sigma^\Luk\not\models_\Luk\chi^\Luk_\cl$. We can now define $\sigma^{\lozenge\Luk}=\!\bigodot_{s\in\Prop(\sigma)}(s\geq\mathbf{1})$. Using Remark~\ref{rem:intervalterms}, one sees that $\sigma^{\lozenge\Luk}\simeq_\Luk\sigma^\Luk$. Thus, $\sigma^{\lozenge\Luk}$ is a~(proper) solution of $\Pmbb^\Luk$.

Conversely, let $\sigma'$ be a~solution of $\Pmbb^\Luk$. It is clear that $\sigma'\not\models_\Luk\chi^\Luk_\cl$ because $\Prop(\chi^\Luk_\cl)\cap\Prop[\Hmsf^\sharp]=\varnothing$, and $\chi^\Luk_\cl$ is not $\Luk$-valid. Hence, $\sigma'$ is a~proper solution. Define $\sigma'^\cl=\bigwedge_{s\in\Prop(\sigma')}\!s$. We show that $\sigma'^\cl$ is a~proper solution of $\Pmbb_\cl$. Since $\sigma'$ is a~solution, we have $\Gamma^\sharp,\sigma'\models_\Luk\chi^\Luk_\cl$. Using Remark~\ref{rem:intervalterms}, we have $\sigma'\simeq_\Luk\!\!\bigodot_{s\in\Prop(\sigma')}\!s$, whence, it is clear that $\Gamma^\sharp,\bigodot_{s\in\Prop(\sigma')}\!s\consvDashLuk\chi^\Luk_\cl$, whence, by Proposition~\ref{prop:CPLtoLuk}, $\Gamma,\sigma'^\cl\consvDashCPL\chi_\cl$, as required.
\end{proof}
\subsection{Relevance and Necessity of Hypotheses\label{ssec:relevance}}
Let us now consider the complexity of determining the relevance and necessity of hypotheses w.r.t.\ solutions to $\Luk$-abduction problems. Namely, given an $\Luk$-abduction problem $\Pmbb=\langle\Gamma,\chi,\Hmsf\rangle$, we will consider the complexity of determining whether an interval literal $\lambda\in\Hmsf$ is relevant (necessary) w.r.t.\ $\Sol(\Pmbb)$, $\PropSol(\Pmbb)$, $\LukminSol(\Pmbb)$, and $\ThminSol(\Pmbb)$.

We begin with the complexity of relevance and necessity w.r.t.\ arbitrary, proper, and $\models_\Luk$-minimal solutions. The classical counterparts of these decision problems were considered by~\cite{EiterGottlob1995}. We show that the complexity of relevance and necessity w.r.t.\ (proper) solutions and $\models_\Luk$-minimal solutions coincides with the complexity of the analogous problems for ($\subseteq$-minimal) solutions in $\CPL$.
\begin{theorem}\label{theorem:relevancecomplexity}
It is $\Sigma^\Pmsf_2$-complete (resp.\ $\Pi^\Pmsf_2$-complete) to decide, given a~$\Luk$-abduction problem $\Pmbb=\langle\Gamma,\chi,\Hmsf\rangle$ and $\lambda\in \Hmsf$, whether $\lambda$ is relevant (resp.\ necessary) w.r.t.\ $\Sol(\Pmbb)$. The same holds for relevance and necessity w.r.t.\ $\PropSol(\Pmbb)$ and $\LukminSol(\Pmbb)$.
\end{theorem}
\begin{proof}
The membership is evident from Theorems~\ref{theorem:arbitrarysolutionrecognitioncomplexity} and~\ref{theorem:Lukminimalsolutionrecognitioncomplexity} as we can just guess a~(proper or $\Luk$-minimal) solution $\tau$ which can be verified in $\DP$ time and then check in linear time whether $\lambda\in\tau$. For the hardness, we adapt the proof by~\cite{EiterGottlob1995} and establish a~reduction from the solution existence \emph{in Łukasiewicz logic}. Now let $\Pmbb=\langle\Gamma,\chi,\Hmsf\rangle$, pick three fresh variables $t$, $t'$, and $t''$, and $\Pmbb^\rmsf=\langle\Gamma^\rmsf,\chi^\rmsf,\Hmsf^\rmsf\rangle$ be as follows:
%a~reduction from the class of $\Luk$-problems presented in~\eqref{equ:solutionexistenceLuk}. Now let $\Pmbb^\Luk=\langle\Gamma^\sharp,\chi^\Luk_\cl,\Hmsf^\sharp\rangle$ be as in~\eqref{equ:solutionexistenceLuk}, pick three fresh variables $t$, $t'$, and $t''$, set $\Xi=\Prop[\Gamma^\sharp\cup\{\chi^\Luk_\cl\}]\cup\{t,t',t''\}$, and $\Pmbb^\rmsf=\langle\Gamma^\rmsf,\chi^\rmsf,\Hmsf^\rmsf\rangle$ be as follows:
\begin{align}
\Gamma^\rmsf&=\{t\!\rightarrow\!\psi\mid\psi\!\in\!\Gamma\}\!\cup\!\{t'\!\rightarrow\!\chi,t\!\rightarrow\!\neg t',t\!\rightarrow\!t'',t'\!\rightarrow\!t''\}\nonumber\\
\chi^\rmsf&=t''\odot\chi\quad
\Hmsf^\rmsf=\Hmsf\cup\{t\geq\mathbf{1},t'\geq\mathbf{1}\}
\label{equ:relevanceLuk}
\end{align}

Now let $\PropSol(\Pmbb)$ be the set of all proper solutions of $\Pmbb$. It is clear that $\PropSol(\Pmbb^\rmsf)=\Sol(\Pmbb^\rmsf)$ (because $t''\notin\Hmsf^\rmsf$) and
\begin{align*}
\Sol(\Pmbb^\rmsf)&=\left\{\!\varrho\!\odot\!(t\!\geq\!\mathbf{1})\!\mid\!\varrho\!\in\!\Sol(\Pmbb^\Luk)\right\}\!\cup\\
&\quad~\left\{\!\varrho'\!\odot\! (t'\geq\mathbf{1})\!\mid\!\exists\Hmsf'\!\subseteq\!\Hmsf\!:\!\varrho'\!=\!\bigodot\limits_{l\in\Hmsf'}\!l\!\right\}
\end{align*}
and that $\Pmbb^\Luk$ has solutions iff $t$ is relevant and $t'$ is not necessary w.r.t.\ $\Sol(\Pmbb^\rmsf)$. 

For hardness w.r.t.\ $\LukminSol(\Pmbb)$, observe that $t$ is relevant to $\Pmbb^\rmsf$ iff it is relevant w.r.t.\ $\models_\Luk$-minimal solutions. Similarly, $t'$ is (not) necessary in $\Pmbb^\rmsf$ iff it is (not) necessary w.r.t.\ $\models_\Luk$-minimal solutions.
\end{proof}

We finish the section by presenting the membership results on the complexity of the recognition of relevant and necessary hypotheses w.r.t.\ theory-minimal solutions. The next statement is an easy consequence of Theorem~\ref{theorem:theoryminimalrecognition}. To the best of our knowledge, no tight complexity bounds have been established for this problem in the $\CPL$-abduction. 
\begin{restatable}{theorem}{theoryrelevancerecognition}\label{theorem:theoryrelevancerecognition}
It is in $\Sigma^\Pmsf_3$ (resp.\ $\Pi^\Pmsf_3$) to decide, given an $\Luk$-ab\-duc\-tion problem $\Pmbb=\langle\Gamma,\chi,\Hmsf\rangle$ and $\lambda\in\Hmsf$, whether $\lambda$ is relevant (resp.\ necessary) w.r.t.\ $\ThminSol(\Pmbb)$.
\end{restatable}
\section{Abduction in Clause Fragments\label{sec:simpleclauseabduction}}
Recall from Proposition~\ref{prop:Luknpcomplete} that the complexity of Łukasiewicz logic coincides with the complexity of $\CPL$. On the other hand, while every classical formula can be equivalently represented as a~set of $\vee$-clauses, $\LLuk$-formulas cannot be transformed into sets of simple clauses. In fact, the simple clause fragment of $\LLuk$ \emph{is decidable in linear time}~\cite[Lemma~2]{BofillManyaVidalVillaret2019}.
\begin{proposition}\label{prop:simpleclausescomplexity}
Let $\Gamma=\{\kappa_1,\ldots,\kappa_n\}$ be a~finite set of simple clauses. It takes linear time to decide whether there is an $\Luk$-valuation $v$ s.t.\ $v(\kappa_i)=1$ for every $i\in\{1,\ldots,n\}$.
\end{proposition}
%KI: Is this true?  SAT for propositional clausal theories is NP-complete, so is the same for simple clauses here?
%DK: Here you can use any real number for the value of a variable. In general (if you don't care for linearity specifically), notice that $v(l_1\oplus\ldots\oplus l_n)=1$ iff $v(l_1)+\ldots+v(l_n)\geq1$. Thus, the satisfiability of a set of simple clauses can be reduced to whether its corresponding system of linear inequalities has a solution. And since we do not introduce integer restrictions, this can be done in polynomial time.

We note briefly that due to~\eqref{equ:demorganproperties}, the following formulas are pairwise strongly equivalent for any $n\in\Nmbb$ and $k<n$:
\begin{align}
\bigoplus\limits^n_{i=1}l_i&&\neg l_1\!\rightarrow\!\bigoplus\limits^n_{i=2}l_i
&&\bigodot^{n-1}_{i=1}\!\neg l_i\!\rightarrow\!l_n
&&\bigodot^k_{i=1}\neg l_i\!\!\rightarrow\!\!\!\!\bigoplus^{n}_{j=n-k}\!\!l_j\label{equ:alternativeclauses}
\end{align}
Proposition~\ref{prop:simpleclausescomplexity} together with~\eqref{equ:alternativeclauses} means, in particular, that the satisfiability of logic programming under Łukasiewicz semantics is polynomial independent of whether it contains negation. In classical logic, the complexity of abduction in polynomial fragments is expectedly lower than in general (cf.~\cite{CreignouZanuttini2006} for details). Hence, it is instructive to establish whether the complexity of $\Luk$-abduction will also be lower in a~polynomial fragment.

In this section, we will consider the complexity of $\Luk$-abduction when the theory is a~set of $\oplus$-clauses. Namely, we will be dealing with two cases: when all clauses in $\Gamma$ are simple and when clauses can contain interval literals. The results are shown in Table~\ref{tab:clausecomplexitysolutions}.
\begin{table}
\centering
\begin{tabular}{lcccc}
&SC&IC&CF&Horn\\\midrule
$\tau\in\Xmsf(\Pmbb)$?&$\Pmsf$&$\DP$&$\Pmsf$&$\Pmsf$\\[.2em]
$\tau\!\in\!\ThminSol(\Pmbb)$?&in $\conp$&in~$\Pi^\Pmsf_2$&in $\conp$&in~$\conp$\\[.2em]
$\Xmsf(\Pmbb)\neq\varnothing$?&$\np$&$\Sigma^\Pmsf_2$&$\np$&$\np$\\[.2em]
rel.\ w.r.t.\ $\Xmsf(\Pmbb)$&$\np$&$\Sigma^\Pmsf_2$&$\np$&$\np$
\end{tabular}
\caption{Abduction in clause fragments. Unless specified otherwise, all results are completeness results. SC --- simple clause theories; CF --- cover-free theories; IC --- theories with arbitrary interval clauses; Horn --- classical Horn theories. $\Xmsf\in\{\Sol,\PropSol,\LukminSol\}$.}
\label{tab:clausecomplexitysolutions}
\end{table}
\subsection{Simple Clause Fragment\label{ssec:simpleclauseabduction}}
Let us now consider the complexity of abductive reasoning for the simple clause fragment of~$\Luk$. We begin with the definition of simple clause abduction problems.
\begin{definition}[Simple clause abduction]\label{def:simpleclauseabduction}
A~\emph{simple clause abduction problem} (SCA problem) is an $\Luk$-abduction problem $\Pmbb=\langle\Gamma,\chi,\Hmsf\rangle$ s.t.\ $\Gamma$ is a~set of simple clauses and interval terms and $\chi$ is an interval clause, simple clause, interval term, or a~simple term.
\end{definition}
In the definition above, note that $\Gamma$ can contain interval terms. This makes it possible to express constraints on the values of variables. Similarly, different observations correspond to constraints that we explain based on the theory.

First, we show that recognition of arbitrary, proper, and $\models_\Luk$-minimal solutions is $\Pmsf$-complete.
\begin{restatable}{theorem}{SCArecognitionpolynomial}\label{theorem:SCArecognitionpolynomial}
Let $\Pmbb=\langle\Gamma,\chi,\Hmsf\rangle$ be an SCA problem, and~$\sigma$ an interval term. Then it is $\Pmsf$-complete to decide whether $\sigma$ is an arbitrary, proper, or $\models_\Luk$-minimal solution.
\end{restatable}
\begin{proof}[Proof sketch]
For the membership, we provide a~sketch of the proof for the case of $\chi$ being an interval literal. Other cases can be dealt with in a~similar manner. Observe from Definitions~\ref{def:Luksemantics} and~\ref{def:intervalterms} that if $\kappa\!=\!\bigoplus^m_{i=1}p_i\oplus\bigoplus^n_{j=1}\!\!\neg q_j$ is an interval clause and $\tau\!=\!\bigodot^r_{i=1}(s_i\lozenge\cmbf_i)$ an interval term, then $v(\kappa)=1$ iff $\sum^m_{i=1}v(p_i)+\sum^n_{j=1}(1-v(q_j))\!\geq\!1$ and $v(\tau)\!=\!1$ iff $v(s_i)\lozenge c_i$ for each $i\in\{1,\ldots,r\}$.
% \begin{align*}
% v(\kappa)=1&\text{ iff }\sum^m_{i=1}v(p_i)+\sum^n_{j=1}(1-v(q_j))\geq1\\
% v(\tau)=1&\text{ iff }\forall i\in\{1,\ldots,r\}:v(s_i)\lozenge c_i
% \end{align*}
It can now be easily shown that the satisfiability of sets of clauses and interval terms can be reduced to solving systems of linear inequalities over $[0,1]$. Similarly, entailment of an interval literal from a set $\Gamma$ containing simple clauses and interval terms can be reduced to verifying that a~system of linear inequalities \emph{does not} have a~solution over $[0,1]$. As both tasks can be done in polynomial time, arbitrary and proper solutions can be recognised in polynomial time. Finally, consider determining whether $\sigma\in\LukminSol(\Pmbb)$. From Theorem~\ref{theorem:Lukminimalsolutionrecognitioncomplexity}, we have that given $\sigma$, we only need to check polynomially many solution candidates, each of which requires polynomial time.

Let us now tackle the hardness. We provide a~logspace reduction from (proper, entailment-minimal) \emph{solution recognition for classical Horn theories}, which is $\Pmsf$-complete. Let $\Pmbb=\langle\Gamma,p,\Hmsf\rangle$ be a~classical Horn abduction problem and $p\notin\Hmsf$, clauses in $\Gamma$ be finite sets of literals, and $\tau\in\Sol(\Pmbb)$. I.e., $\Gamma\cup\{\tau\}$ is $\CPL$-satisfiable, and there is a~unit resolution inference of $\{p\}$ from $\Gamma\cup\{\tau\}$ (i.e., in each application of the resolution rule, at least one premise is a~unit clause). Let $\Gamma^\oplus$ be the result of replacing each clause $\{l_1,\ldots,l_k\}$ in $\Gamma$ with $l_1\oplus\ldots\oplus l_k$; $\tau^{\lozenge\Luk}$ be the result of replacing $\wedge$ with $\odot$, positive literals $r$ in $\tau$ with $r\!\geq\!\mathbf{1}$, and negative literals $\neg q$ with $q\leq\mathbf{0}$; and $\Hmsf^\Luk$ be the result of replacing $r$ in $\Hmsf$ with $r\!\geq\mathbf{1}$, and $\neg q$ with $q\!\leq\!\mathbf{0}$. The size of $\Pmbb^\oplus=\langle\Gamma^\oplus,p,\Hmsf^\Luk\rangle$ is linear in the size of $\Pmbb$, and $\Gamma^\oplus\cup\{\tau^{\lozenge\Luk}\}$ is $\Luk$-satisfiable because $\Gamma\cup\{\tau\}$ is $\CPL$-satisfiable. Furthermore, unit resolution is $\Luk$-sound: if $v(l_1\oplus\ldots\oplus l_k)=1$ and $v(l_1)=0$, then $v(l_2\oplus\ldots\oplus l_k)=1$. Hence, reusing the inference in~$\Luk$, we have $\Gamma^\oplus\cup\{\tau^{\lozenge\Luk}\}\models_\Luk p$. Thus, $\tau^{\lozenge\Luk}\in\Sol(\Pmbb^\oplus)$. Conversely, let $\tau$ solve $\Pmbb^\oplus$, i.e., $\Gamma^\oplus\cup\{\tau^{\lozenge\Luk}\}\consvDashLuk p$. Let $\tau_\cl$ be obtained from $\tau$ by replacing $\odot$ with $\wedge$, $r\!\geq\!\mathbf{1}$ with $r$, and $q\!\leq\!\mathbf{0}$ with $\neg q$. Clearly, $\Gamma,\tau_\cl\models_\CPL p$. Assume for contradiction that $\Gamma,\tau_\cl\models_\CPL\bot$. Then there is a classical unit resolution inference of the empty clause from $\Gamma\cup\{\tau_\cl\}$. As it is sound in~$\Luk$, $\Gamma,\tau_\cl\models_\Luk\bot$. Contradiction.

Observe that all solutions of $\Pmbb$ and $\Pmbb^\oplus$ \emph{are proper} because $p$ does not occur in them. Thus, our reduction can be used to show that determining the existence of proper solutions is $\Pmsf$-hard as well. For the $\Pmsf$-hardness of $\models_\Luk$-minimal solution recognition, one can check that if $\tau\in\LukminSol(\Pmbb)$, then $\tau^{\lozenge\Luk}\in\LukminSol(\Pmbb^\oplus)$. Conversely, if $\sigma\in\LukminSol(\Pmbb^\oplus)$, then $\sigma_\cl\in\LukminSol(\Pmbb)$.
% Let $\kappa=\bigoplus^m_{i=1}p_i\oplus\bigoplus^n_{j=1}\neg q_j$ be a~simple clause, $\tau=\bigodot^r_{i=1}(s_i\lozenge\cmbf_i)$ an interval term, and $\kappa'=\bigoplus^k_{i=1}(t_i\lozenge\cmbf_i)$ an interval clause. Observe from Definitions~\ref{def:Luksemantics} and~\ref{def:intervalterms} that 
% \begin{align*}
% v(\kappa)=1&\text{ iff }\sum^m_{i=1}v(p_i)+\sum^n_{j=1}(1-v(q_j))\geq1\nonumber\\
% v(\tau)=1&\text{ iff }\forall i\in\{1,\ldots,r\}:v(s_i)\lozenge c_i\nonumber\\
% v(\kappa')=1&\text{ iff }\exists i\in\{1,\ldots,k\}:v(t_i)\lozenge c_i
% \end{align*}
% It can now be easily shown that the satisfiability of sets of clauses and interval terms can be reduced to solving systems of linear inequalities over $[0,1]$. Similarly, entailment of interval terms, interval clauses, and simple clauses from such sets of formulas can be reduced to verifying that a~system of linear inequalities \emph{does not} have a~solution over $[0,1]$. As both tasks can be done in polynomial time, arbitrary and proper solutions can be recognised in polynomial time.
% Finally, consider determining whether $\sigma\in\LukminSol(\Pmbb)$. From Theorem~\ref{theorem:Lukminimalsolutionrecognitioncomplexity}, we have that given $\sigma$, we only need to check polynomially many candidate solutions, each of which requires polynomial time. The result follows.
\end{proof}

For the recognition of \emph{theory-minimal} solutions, we provide a~$\conp$ membership result.
\begin{restatable}{theorem}{SCAtheoryminimalrecognitionconp}\label{theorem:SCAtheoryminimalrecognitionconp}
It is in $\conp$ to decide, given an SCA problem~$\Pmbb$ and an interval term~$\tau$, whether $\tau$ is a~the\-o\-ry-minimal solution of $\Pmbb$.
\end{restatable}

As expected, \emph{solution existence} for simple clause abduction problems is $\np$-complete.
\begin{restatable}{theorem}{SCAsolutionexistencenp}\label{theorem:SCAsolutionexistencenp}
Given a~simple clause abduction problem $\Pmbb$, it is $\np$-complete to decide whether it has (proper) solutions.
\end{restatable}
\begin{proof}
The membership is immediate from Theorem~\ref{theorem:SCArecognitionpolynomial}. For the hardness, we can reuse the reduction from classical Horn abduction given in Theorem~\ref{theorem:SCArecognitionpolynomial}.
\end{proof}

The $\np$-completeness of relevance of hypotheses can be obtained using the reduction shown in~\eqref{equ:relevanceLuk}, setting $\chi=p$ and assuming that $\Gamma$ is a~set of simple clauses. Indeed, if $\psi$ is a~simple clause, then $t\rightarrow\psi$ can be represented as a~simple clause (recall~\eqref{equ:alternativeclauses}).
\begin{theorem}\label{theorem:SCArelevance}
It is $\np$-complete (resp.\ $\conp$-complete) to decide, given an SCA problem $\Pmbb=\langle\Gamma,\chi,\Hmsf\rangle$ and $\lambda\in \Hmsf$, whether $\lambda$ is relevant (resp.\ necessary) w.r.t.\ $\Sol(\Pmbb)$. The same holds for relevance and necessity w.r.t.\ $\PropSol(\Pmbb)$ and $\LukminSol(\Pmbb)$.
\end{theorem}

We finish the section with two remarks. First, the complexity of simple clause abduction coincides with that of \emph{classical Horn abduction}~\cite{EiterGottlob1995,CreignouZanuttini2006}. Thus, Łukasiewicz abduction with clausal theories is simpler than classical abduction over clausal theories. Second, one can easily see that simple clause abduction problems can be straightforwardly generalised to problems whose theories are \emph{Łukasiewicz fuzzy logic programmes} as presented by~\cite{Vojtas1999,Vojtas2001} \emph{while preserving the complexity}.%
\begin{definition}\label{def:FLPAP}
A~\emph{Łukasiewicz fuzzy logic programme} ($\Luk$-FLP) is a~finite set $\Gamma_\Pmsf=\{\langle\kappa_i,x_i\rangle\mid 1\!\leq\!i\!\leq\!n, x_i\!\in\!(0,1]_\Qmbb\}$ with $\kappa_i$'s being simple clauses written as $\bigodot^m_{i=1}l_i\rightarrow l$. Pairs $\langle\kappa_i,x_i\rangle$ are called \emph{fuzzy rules}. An $\Luk$-valuation $v$ \emph{satisfies $\Gamma_\Pmsf$} if $v(\kappa_i)=x_i$ for every $\langle\kappa_i,x_i\rangle\in\Gamma_\Pmsf$.

An \emph{$\Luk$-FLP abduction problem} is a~tuple $\Pmbb=\langle\Gamma_\Pmsf,\chi,\Hmsf\rangle$ with $\Gamma_\Pmsf$ being an~$\Luk$-FLP, $\chi$ a~fuzzy rule or interval term, and $\Hmsf$ a~set of interval literals. A~solution to $\Pmbb$ is defined as in Definition~\ref{def:Lukabduction}.
\end{definition}

One can observe from Definition~\ref{def:simpleclauseabduction} that simple clause abduction problems are a~particular case of $\Luk$-FLP abduction problems (namely, when $x_i=1$ for every~$i$). Thus, the hardness results are preserved. For the membership, note that $\Luk$-FLP abduction problems can be reduced to solving systems of linear inequalities in the same way as simple clause abduction problems.
\subsection{Interval Clause Fragment\label{ssec:intervalclausesabduction}}
Results in Section~\ref{sec:Lukabduction} show that the complexity of $\Luk$-abduction in the general case was not affected by whether we allow the use of interval literals \emph{in theories or observations} (recall from Definition~\ref{def:Lukabduction} that theories and observations are defined over~$\LLukint$). Interval literals can also reduce the size and facilitate the understanding (for a~human) of an abduction problem %w.r.t.\ the formulation in~$\LLuk$
%simplify the formalisation of abduction problems
while preserving its solutions. E.g., the theory $\Gamma_\mathsf{lift}$ of the problem $\Pmbb_\mathsf{lift}$ from Examples~\ref{example:thresholds} and~\ref{example:thresholdssolution} can be reformulated as follows: $\Gamma^\Qmbb_\mathsf{lift}=\{(c\geq\mathbf{\tfrac{1}{4}})\leftrightarrow g,(c\leq\mathbf{\tfrac{2}{3}})\leftrightarrow b\}$. Observe that $(c\!\geq\!\mathbf{\tfrac{1}{4}})\!\leftrightarrow\!g$ is weakly equivalent to $(c\!\oplus\!c\!\oplus\!c\!\oplus\!c)\!\leftrightarrow\!g$ and $(c\!\leq\!\mathbf{\tfrac{2}{3}})\!\leftrightarrow\!b$ to $(\neg c\!\oplus\!\neg c\!\oplus\!\neg c)\!\leftrightarrow\!b$.
% \begin{align*}
% (c\geq\mathbf{\sfrac{1}{4}})\rightarrow g&\simeq_\Luk(c\oplus c\oplus c\oplus c)\rightarrow g\\
% (c\leq\mathbf{\sfrac{2}{3}})\rightarrow b&\simeq_\Luk(\neg c\oplus\neg c\oplus\neg c)\rightarrow b
% \end{align*}
So, for $\Pmbb^\Qmbb_\mathsf{lift}=\langle\Gamma^\Qmbb_\mathsf{lift},g\odot b,\Hmsf_\mathsf{lift}\rangle$, we have $\Sol(\Pmbb_\mathsf{lift})=\Sol(\Pmbb^\Qmbb_\mathsf{lift})$ (and likewise for other types of solutions).
% Moreover, fuzzy rules of the form $\langle\phi,c\rangle$ for arbitrary $\phi\in\LLukint$ can be simulated using interval terms as $\phi_{=c}\coloneqq(\phi\leftrightarrow p)\odot(p\geq\cmbf)\odot(p\leq\cmbf)$. Indeed, $v(\phi_{=c})=1$ iff $v(\phi)=c$ by Definitions~\ref{def:Luksemantics} and~\ref{def:intervalterms}.
Moreover, using interval literals in $\oplus$-clauses simplifies the presentation of relations between values. E.g., ‘if $v(p)\geq\tfrac{1}{2}$, then $v(q)\!\leq\!\tfrac{3}{4}$’ can be put as $(p\!\geq\!\mathbf{\tfrac{1}{2}})\!\rightarrow\!(q\!\leq\!\mathbf{\tfrac{3}{4}})$.

Thus, one might wonder whether we can permit theories built from \emph{interval clauses} (recall Definition~\ref{def:intervalterms}) and preserve the complexity bounds from Section~\ref{ssec:simpleclauseabduction}. As the next statements show, this is not generally the case.
\begin{restatable}{theorem}{intervalclauseabductioncomplexityrecognition}\label{theorem:intervalclauseabductioncomplexityrecognition}
Let $\Pmbb\!=\!\langle\Gamma,p,\Hmsf\rangle$ be an $\Luk$-abduction problem s.t.\ $\Gamma$ is a~set of interval clauses and $p\!\in\!\Prop$. Then, for an interval term $\tau$, it is $\DP$-complete to decide whether \mbox{$\tau\in\Sol(\Pmbb)$.}
\end{restatable}
\begin{proof}[Proof sketch]
Membership follows from Theorem~\ref{theorem:arbitrarysolutionrecognitioncomplexity}. Hardness can be shown via a~reduction from solution recognition. in~$\CPL$. Let $\Pmbb=\langle\Gamma,p,\Hmsf\rangle$ be a~classical abduction problem with $\Gamma$ a~set of $\vee$-clauses. We replace all positive literals $r$ occurring in~$\Pmbb$ with $r\geq\mathbf{1}$ and all negative literals $\neg s$ with $s<\mathbf{1}$. Given~$\tau\in\Sol(\Pmbb)$, we produce an $\Luk$-solution by a~similar replacement of literals and changing $\wedge$ to $\odot$. Conversely, a~classical solution for~$\Pmbb$ can be obtained from a~solution for the $\Luk$-abduction problem by a~reverse replacement.
% Assume that $\Pmbb=\langle\Gamma,p,\Hmsf\rangle$ is a~\emph{classical abduction problem} with $\Gamma=\{\kappa_i\mid i\in\{1,\ldots,n\}\}$ being a~set of disjunctive clauses and that $\tau$ is a~solution to $\Pmbb$. We define $\Pmbb_\interval=\langle\Gamma_\interval,p,\Hmsf_\interval\rangle$ and $\tau_\interval$ as follows:
% \begin{align}\label{equ:intervalclausereduction}
% \Gamma_\interval&=\left\{\bigoplus\limits_{q_i\in\kappa_i}(q_i\!\geq\!\mathbf{1})\oplus\bigoplus\limits_{\neg r_i\in\kappa_i}(r_i\!<\!\mathbf{1})\mid i\in\{1,\ldots,n\}\right\}\nonumber\\
% \Hmsf_\interval&=\{s\geq\mathbf{1}\mid s\in\Hmsf\}\cup\{t<\mathbf{1}\mid\neg t\in\Hmsf\}\nonumber\\
% \tau_\interval&=\bigodot\limits_{s'\in\tau}(s'\geq\mathbf{1})\odot\bigodot\limits_{\neg t'\in\tau}(t<\mathbf{1})
% \end{align}
% Now let $\tau\!\in\!\Sol(\Pmbb)$, i.e., $\Gamma,\tau\!\consvDashCPL\!p$. Clearly, $\Gamma_\interval,\tau_\interval\!\not\models_\Luk\!\bot$. For the entailment, assume for contradiction that $v$ witnesses $\Gamma_\interval,\tau_\interval\not\models_\Luk p$. Consider a classical valuation $v_\cl$ s.t. $v_\cl(q)=1$ iff $v(q)=1$. One can easily see that $v_\cl$ witnesses $\Gamma,\tau\not\models_\CPL p$. For the converse direction, let $\tau\in\Sol(\Pmbb_\interval)$. A~solution for $\Pmbb$ can be obtained by replacing $r\geq\mathbf{1}$ with $r$ and $r<1$ with $\neg r$. The reasoning is similar.
\end{proof}

The $\Sigma^\Pmsf_2$-hardness of solution existence for $\Luk$-abduction in the interval clause fragment can be obtained by a~reduction used in Theorem~\ref{theorem:intervalclauseabductioncomplexityrecognition}.
\begin{restatable}{theorem}{intervalclauseabductioncomplexityexistence}\label{theorem:intervalclauseabductioncomplexityexistence}
Let $\Pmbb=\langle\Gamma,l,\Hmsf\rangle$ be an $\Luk$-abduction problem s.t.\ $\Gamma$ is a~set of interval clauses. Then it is $\Sigma^\Pmsf_2$-complete to decide whether there is a~(proper) solution to~$\Pmbb$.
\end{restatable}

As expected, determining the relevance of $\lambda\in\Hmsf$ is $\Sigma^\Pmsf_2$-hard. The next statement can be obtained by a~reduction shown in~\eqref{equ:relevanceLuk}. The only difference is that instead of \emph{variables} $t$, $t'$, $t''$, we take \emph{interval literals} $t\geq\mathbf{1}$, $t'\geq\mathbf{1}$, and $t''\geq\mathbf{1}$.
\begin{theorem}\label{theorem:ICrelevancecomplexity}
It is $\Sigma^\Pmsf_2$-complete (resp.\ $\Pi^\Pmsf_2$-complete) to decide, given a~$\Luk$-abduction problem $\Pmbb=\langle\Gamma,\chi,\Hmsf\rangle$ s.t.\ $\Gamma$ is a~set of interval clauses and $\lambda\in\Hmsf$, whether $\lambda$ is relevant (resp.\ necessary) w.r.t.\ $\Sol(\Pmbb)$. The same holds for relevance and necessity w.r.t.\ $\PropSol(\Pmbb)$ and $\LukminSol(\Pmbb)$.
\end{theorem}

Still, we can restrict interval clause theories so that abduction becomes $\np$-complete. Namely, we prohibit interval literals ‘covering’ the entire $[0,1]$ interval (e.g., $p\!\leq\!\mathbf{\tfrac{2}{3}}$ and $p\!\geq\!\mathbf{\tfrac{1}{4}}$) in the implicative representation of interval clauses.% This restriction replicates the behaviour of \emph{classical Horn theories}.
\begin{definition}\label{def:coverfreeintervalclauses}
Let $\Gamma$ be a~set of interval clauses represented as $\bigodot^m_{i=1}(p_i\lozenge\cmbf_i)\!\rightarrow\!(q\lozenge\dmbf)$ and $\bigodot^m_{i=1}(p_i\lozenge\cmbf_i)\!\rightarrow\!\bot$. $\Gamma$ is \emph{cover-free} (CF) if no pair of interval literals $\lambda$ and $\lambda'$ over the same variable occurs in~$\Gamma$ s.t.\ $\Luk\models\lambda\oplus\lambda'$.
\end{definition}

The following theorems can be shown by reductions between Horn abduction problems and $\Luk$-abduction problems with CF theories.
\begin{restatable}{theorem}{coverfreesolutionrecognition}\label{theorem:coverfreesolutionrecognition}
Let $\Pmbb=\langle\Gamma,\tau,\Hmsf\rangle$ be an $\Luk$-abduction problem with $\Gamma$ a~CF set of interval clauses and $\tau$ an interval term. Then it is $\Pmsf$-complete to check whether an interval term $\sigma$ is a~(proper, minimal) solution to~$\Pmbb$.
% \begin{align*}
% \sigma\in\Sol(\Pmbb)&&\sigma\in\PropSol(\Pmbb)&&\sigma\in\LukminSol(\Pmbb)
% \end{align*}
\end{restatable}
\begin{restatable}{theorem}{coverfreesolutionexistence}\label{theorem:coverfreesolutionexistence}
Let $\Pmbb=\langle\Gamma,\tau,\Hmsf\rangle$ be an $\Luk$-abduction problem with $\Gamma$ a~CF set of interval clauses. Then it is $\np$-complete to decide whether $\Pmbb$ has (proper) solutions.
\end{restatable}

Determining the relevance of hypotheses in CF theories is $\np$-complete. The proof is the same as for Theorem~\ref{theorem:ICrelevancecomplexity} (but we reduce from hypotheses relevance for Horn theories).
\begin{theorem}\label{theorem:CFrelevance}
It is $\np$-complete (resp.\ $\conp$-complete) to decide, given an $\Luk$-abduction problem $\Pmbb=\langle\Gamma,\chi,\Hmsf\rangle$ with $\Gamma$~being a~CF set of interval clauses and $\lambda\in\Hmsf$, whether $\lambda$ is relevant (resp.\ necessary) w.r.t.\ $\Sol(\Pmbb)$. The same holds for relevance and necessity w.r.t.\ $\PropSol(\Pmbb)$ and $\LukminSol(\Pmbb)$.
\end{theorem}

The next statement follows from Theorem~\ref{theorem:coverfreesolutionrecognition} as for interval terms $\sigma$ and $\tau$ and a~CF theory $\Gamma$, it takes polynomial time to check whether $\Gamma,\sigma\models_\Luk\tau$.
\begin{theorem}\label{theorem:CFtheoryminimalinconp}
Let $\Pmbb=\langle\Gamma,\tau,\Hmsf\rangle$ be an $\Luk$-abduction problem with $\Gamma$ a~CF theory and $\tau$ an interval term. Then given an interval term $\sigma$, it is in $\conp$ to decide whether $\sigma\in\ThminSol(\Pmbb)$.
\end{theorem}
\section{Conclusion and Discussion\label{sec:conclusion}}
% Let us summarise the results of the paper.
We studied abduction in Łukasiewicz logic and its clausal fragments. Our analysis gives an almost complete outline of the complexity of the main decision problems related to abduction (Tables~\ref{tab:complexitysolutions} and~\ref{tab:clausecomplexitysolutions}). We established that the complexity of abductive reasoning is never higher than that in $\CPL$~\cite{EiterGottlob1995}, \cite{CreignouZanuttini2006}, \cite{PichlerWoltran2010}, \cite{PfandlerPichlerWoltran2015}. Moreover, the complexity of $\Luk$-abduction in clausal fragment is \emph{lower} than that of the classical abductive reasoning as long as clauses do not contain interval literals. 

Several questions remain open. First, we do not know the exact complexity of theory-minimal solution recognition and relevance in the full language nor its clausal fragments. One way to approach this would be to establish the complexity of the closely related notion of theory prime implicants in $\CPL$ and its fragments~\cite{Marquis1995}.

It would also be instructive to consider abductive reasoning in fuzzy logics when the entailment is defined via the preservation of the \emph{truth degree} from the premise to the conclusion as proposed by~\cite{BouEstevaFontGilGodoTorrensVerdu2009,ErtolaEstevaFlaminioGodoNoguera2015}. Indeed, one might not be guaranteed that the statements in the theory are \emph{absolutely true}. Still, in this case, it is reasonable to expect that the observation should be \emph{at least as true} as the theory and the explanation.

Furthermore, \cite{DuboisPrade1992}, \cite{Poole1993}, \cite{DuboisGilioKern-Isberner2008}, and~\cite{SatoIshihataInoue2011} have considered abduction in possibilistic and probabilistic contexts. It is also known from~\cite{HajekTulipani2001,BaldiCintulaNoguera2020} that so-called ‘two-layered’ fuzzy logics can be adapted to reasoning about uncertainty. Thus, it makes sense to explore probabilistic abduction using two-layered logics.

We also note that the context from Example~\ref{example:thresholds} can be represented in abductive constraint logic programming (ACLP) studied by~\cite{KakasMichaelMourlas2000}. It thus makes sense to explore the connection between abduction in fuzzy logic and ACLP.

It is also known that machine learning for visual perception problems is related to abduction~\cite{Shanahan2005,LiangWangZhouYang2022}. As fuzzy logic has also found numerous applications in visual perception, it makes sense to explore the applications of abduction in fuzzy logic in this field.

Finally, recall that there are algorithms for solving fuzzy logic abduction problems with fuzzy sets. %In general, however, there are infinitely many such solutions for a~given problem. On the other hand, abduction problems can have only finitely many solutions for a~given finite set of interval literals.
In our framework, however, a~solution corresponds to a~set of \emph{intervals} of permitted values of hypotheses. As $\Luk$-entailment is reducible to \emph{finite-valued} Łukasiewicz logics~\cite{AguzzoliCiabattoni2000}, we might hope that it is possible to compute the ‘minimally sufficient’ set of permitted values from the shape of the theory. Thus, it would be instructive to devise algorithms for the generation of solutions in the form of interval terms. As satisfiability and validity in~$\Luk$ are reducible to mixed-integer programming (MIP)~\cite{Haehnle1992,Haehnle1994}, it would make sense to apply  MIP solvers. Moreover, since the semantics of~$\Luk$ can be given over \emph{rational numbers}~\cite{EstevaGispertGodoMontagna2002}, one may try to apply \emph{rational} MIP solvers such as the one by~\cite{CookKochSteffyWolter2013} as they produce \emph{exact} solutions to mixed-integer problems.
\section*{Acknowledgements}
The $1^\mathsf{st}$ author was supported by JSPS KAKENHI Grant Number JP25K03190 and JST CREST Grant Number JPMJCR22D3, Japan. We would also like to thank the reviewers for their constructive suggestions and spotting a~mistake in the first version of the paper.
\bibliographystyle{kr}
\bibliography{kr-sample}

\begin{thebibliography}{}

\bibitem[\protect\citeauthoryear{Aguzzoli and Ciabattoni}{2000}]{AguzzoliCiabattoni2000}
Aguzzoli, S., and Ciabattoni, A.
\newblock 2000.
\newblock Finiteness in infinite-valued {\l}ukasiewicz logic.
\newblock {\em Journal of Logic, Language and Information} 9(1):5--29.

\bibitem[\protect\citeauthoryear{Baaz}{1996}]{Baaz1996}
Baaz, M.
\newblock 1996.
\newblock {Infinite-valued G{\"o}del logics with $0 $-$1 $-projections and relativizations}.
\newblock In {\em G{\"o}del'96: Logical foundations of mathematics, computer science and physics---Kurt G{\"o}del's legacy, Brno, Czech Republic, August 1996, proceedings}. Association for Symbolic Logic.
\newblock  23--33.

\bibitem[\protect\citeauthoryear{Badreddine \bgroup et al\mbox.\egroup }{2022}]{BadreddinedAvila-GarcezSerafiniSpranger2022}
Badreddine, S.; Avila~Garcez, A.; Serafini, L.; and Spranger, M.
\newblock 2022.
\newblock Logic tensor networks.
\newblock {\em Artificial Intelligence} 303(103649):103649.

\bibitem[\protect\citeauthoryear{Baldi, Cintula, and Noguera}{2020}]{BaldiCintulaNoguera2020}
Baldi, P.; Cintula, P.; and Noguera, C.
\newblock 2020.
\newblock {Classical and Fuzzy Two-Layered Modal Logics for Uncertainty: Translations and Proof-Theory}.
\newblock {\em International Journal of Computational Intelligence Systems} 13(1):988--1001.

\bibitem[\protect\citeauthoryear{Bergadano, Cutello, and Gunetti}{2000}]{BergadanoCutelloGunetti2000}
Bergadano, F.; Cutello, V.; and Gunetti, D.
\newblock 2000.
\newblock Abduction in machine learning.
\newblock In Gabbay, D., and Smets, P., eds., {\em Abductive Reasoning and Learning}, volume~4 of {\em Handbook of Defeasible Reasoning and Uncertainty Management Systems}. Dordrecht: Springer Netherlands.
\newblock  197--229.

\bibitem[\protect\citeauthoryear{Bhagavatula \bgroup et al\mbox.\egroup }{2020}]{BhagavatulaLeBrasMalaviyaSakaguchiHoltzmanRashkinDowneyYihChoi2020}
Bhagavatula, C.; Le~Bras, R.; Malaviya, C.; Sakaguchi, K.; Holtzman, A.; Rashkin, H.; Downey, D.; Yih, W.-t.; and Choi, Y.
\newblock 2020.
\newblock {Abductive Commonsense Reasoning}.
\newblock In {\em 8th International Conference on Learning Representations, {ICLR} 2020}.
\newblock OpenReview.net.

\bibitem[\protect\citeauthoryear{Bofill \bgroup et al\mbox.\egroup }{2019}]{BofillManyaVidalVillaret2019}
Bofill, M.; Many{\`a}, F.; Vidal, A.; and Villaret, M.
\newblock 2019.
\newblock New complexity results for {\l}ukasiewicz logic.
\newblock {\em Soft Comput.} 23(7):2187--2197.

\bibitem[\protect\citeauthoryear{Borgwardt and Pe{\~{n}}aloza}{2012}]{BorgwardtPenaloza2012}
Borgwardt, S., and Pe{\~{n}}aloza, R.
\newblock 2012.
\newblock Undecidability of fuzzy description logics.
\newblock In Brewka, G.; Eiter, T.; and McIlraith, S., eds., {\em Principles of Knowledge Representation and Reasoning: Proceedings of the Thirteenth International Conference}.
\newblock {AAAI} Press.

\bibitem[\protect\citeauthoryear{Borgwardt and Pe{\~n}aloza}{2017}]{BorgwardtPenaloza2017}
Borgwardt, S., and Pe{\~n}aloza, R.
\newblock 2017.
\newblock Algorithms for reasoning in very expressive description logics under infinitely valued {G}{\"o}del semantics.
\newblock {\em International Journal of Approximate Reasoning} 83:60--101.

\bibitem[\protect\citeauthoryear{Borgwardt, Distel, and Pe{\~n}aloza}{2014}]{BorgwardtDistelPenaloza2014DL}
Borgwardt, S.; Distel, F.; and Pe{\~n}aloza, R.
\newblock 2014.
\newblock G{\"o}del description logics with general models.
\newblock In {\em DL 2014: Informal Proceedings of the 27th International Workshop on Description Logics, Vienna, Austria, July 17--20, 2014}, volume 1193,  391--403.
\newblock CEUR.

\bibitem[\protect\citeauthoryear{Borgwardt}{2014}]{Borgwardt2014PhD}
Borgwardt, S.
\newblock 2014.
\newblock {\em Fuzzy Description Logics with General Concept Inclusions}.
\newblock Ph.D. Dissertation, Technische Universität Dresden, Dresden.

\bibitem[\protect\citeauthoryear{Bou \bgroup et al\mbox.\egroup }{2009}]{BouEstevaFontGilGodoTorrensVerdu2009}
Bou, F.; Esteva, F.; Font, J.; Gil, A.; Godo, L.; Torrens, A.; and Verdu, V.
\newblock 2009.
\newblock Logics preserving degrees of truth from varieties of residuated lattices.
\newblock {\em Journal of Logic and Computation} 19(6):1031--1069.

\bibitem[\protect\citeauthoryear{Chakraborty \bgroup et al\mbox.\egroup }{2013}]{ChakrabortyKonarPalJain2013}
Chakraborty, A.; Konar, A.; Pal, N.; and Jain, L.
\newblock 2013.
\newblock Extending the contraposition property of propositional logic for fuzzy abduction.
\newblock {\em IEEE Transactions in Fuzzy Systems} 21(4):719--734.

\bibitem[\protect\citeauthoryear{Chen, Hu, and Sun}{2022}]{ChenHuSun2022}
Chen, X.; Hu, Z.; and Sun, Y.
\newblock 2022.
\newblock Fuzzy logic based logical query answering on knowledge graphs.
\newblock In {\em Proceedings of the 36th AAAI Conference on Artificial Intelligence}, volume~36,  3939--3948.
\newblock Association for the Advancement of Artificial Intelligence (AAAI).

\bibitem[\protect\citeauthoryear{Cook \bgroup et al\mbox.\egroup }{2013}]{CookKochSteffyWolter2013}
Cook, W.; Koch, T.; Steffy, D.; and Wolter, K.
\newblock 2013.
\newblock A hybrid branch-and-bound approach for exact rational mixed-integer programming.
\newblock {\em Mathematical Programming Computation} 5(3):305--344.

\bibitem[\protect\citeauthoryear{Creignou and Zanuttini}{2006}]{CreignouZanuttini2006}
Creignou, N., and Zanuttini, B.
\newblock 2006.
\newblock {A Complete Classification of the Complexity of Propositional Abduction}.
\newblock {\em SIAM Journal of Computing} 36(1):207--229.

\bibitem[\protect\citeauthoryear{Dai \bgroup et al\mbox.\egroup }{2019}]{DaiXuYuZhou2019}
Dai, W.-Z.; Xu, Q.; Yu, Y.; and Zhou, Z.-H.
\newblock 2019.
\newblock Bridging machine learning and logical reasoning by abductive learning.
\newblock {\em Advances in Neural Information Processing Systems} 32.

\bibitem[\protect\citeauthoryear{d'Allonnes, Akdag, and Bouchon-Meunier}{2007}]{dAllonnesAkdagBouchon-Meunier2007}
d'Allonnes, A.; Akdag, H.; and Bouchon-Meunier, B.
\newblock 2007.
\newblock Selecting implications in fuzzy abductive problems.
\newblock In {\em 2007 {IEEE} Symposium on Foundations of Computational Intelligence}.
\newblock IEEE.

\bibitem[\protect\citeauthoryear{Diligenti, Gori, and Sacc{\`a}}{2017}]{DiligentiGoriSacca2017}
Diligenti, M.; Gori, M.; and Sacc{\`a}, C.
\newblock 2017.
\newblock Semantic-based regularization for learning and inference.
\newblock {\em Artificial Intelligence} 244:143--165.

\bibitem[\protect\citeauthoryear{Dubois and Prade}{1992}]{DuboisPrade1992}
Dubois, D., and Prade, H.
\newblock 1992.
\newblock {Possibilistic Abduction}.
\newblock In {\em International Conference on Information Processing and Management of Uncertainty in Knowledge-Based Systems},  3--12.
\newblock Springer.

\bibitem[\protect\citeauthoryear{Dubois, Gilio, and Kern-Isberner}{2008}]{DuboisGilioKern-Isberner2008}
Dubois, D.; Gilio, A.; and Kern-Isberner, G.
\newblock 2008.
\newblock {Probabilistic Abduction Without Priors}.
\newblock {\em International Journal of Approximate Reasoning} 47(3):333--351.

\bibitem[\protect\citeauthoryear{Ebrahim}{2001}]{Ebrahim2001}
Ebrahim, R.
\newblock 2001.
\newblock Fuzzy logic programming.
\newblock {\em Fuzzy Sets And Systems} 117(2):215--230.

\bibitem[\protect\citeauthoryear{Eiter and Gottlob}{1995}]{EiterGottlob1995}
Eiter, T., and Gottlob, G.
\newblock 1995.
\newblock {The Complexity of Logic-Based Abduction}.
\newblock {\em Journal of the Association for Computing Machinery} 42(1):3--42.

\bibitem[\protect\citeauthoryear{El~Ayeb, Marquis, and Rusinowitch}{1993}]{ElAyebMarquisRusinowitch1993}
El~Ayeb, B.; Marquis, P.; and Rusinowitch, M.
\newblock 1993.
\newblock Preferring diagnoses by abduction.
\newblock {\em IEEE Transactions on Systems, Man, and Cybernetics} 23(3):792--808.

\bibitem[\protect\citeauthoryear{Ertola \bgroup et al\mbox.\egroup }{2015}]{ErtolaEstevaFlaminioGodoNoguera2015}
Ertola, R.; Esteva, F.; Flaminio, T.; Godo, L.; and Noguera, C.
\newblock 2015.
\newblock Paraconsistency properties in degree-preserving fuzzy logics.
\newblock {\em Soft Computing} 19:531--546.

\bibitem[\protect\citeauthoryear{Esteva \bgroup et al\mbox.\egroup }{2002}]{EstevaGispertGodoMontagna2002}
Esteva, F.; Gispert, J.; Godo, L.; and Montagna, F.
\newblock 2002.
\newblock {On the Standard and Rational Completeness of Some Axiomatic Extensions of the Monoidal T-norm Logic}.
\newblock {\em Studia Logica} 71(2):199--226.

\bibitem[\protect\citeauthoryear{Flach and Kakas}{2000}]{FlachKakas2000}
Flach, P., and Kakas, A., eds.
\newblock 2000.
\newblock {\em Abduction and Induction}, volume~18 of {\em Applied Logic Series}.
\newblock Dordrecht, Netherlands: Springer.

\bibitem[\protect\citeauthoryear{Flaminio}{2007}]{Flaminio2007}
Flaminio, T.
\newblock 2007.
\newblock {NP}-containment for the coherence test of assessments of conditional probability: a fuzzy logical approach.
\newblock {\em Archive for Mathematical Logic} 46(3–4):301--319.

\bibitem[\protect\citeauthoryear{H{\"a}hnle}{1992}]{Haehnle1992}
H{\"a}hnle, R.
\newblock 1992.
\newblock {A New Translation From Deduction Into Integer Programming}.
\newblock In {\em International Conference on Artificial Intelligence and Symbolic Mathematical Computing},  262--275.
\newblock Springer.

\bibitem[\protect\citeauthoryear{H{\"a}hnle}{1994}]{Haehnle1994}
H{\"a}hnle, R.
\newblock 1994.
\newblock {Many-Valued Logic and Mixed Integer Programming}.
\newblock {\em Annals of Mathematics and Artificial Intelligence} 12(3-4):231--263.

\bibitem[\protect\citeauthoryear{H\"{a}hnle}{1999}]{Haehnle1999}
H\"{a}hnle, R.
\newblock 1999.
\newblock {Tableaux for Many-Valued Logics}.
\newblock In D’Agostino, M.; Gabbay, D.; H\"{a}hnle, R.; and Posegga, J., eds., {\em Handbook of Tableaux Methods},  529--580.
\newblock Springer-Science+Business Media, B.V.

\bibitem[\protect\citeauthoryear{H{\'a}jek and Tulipani}{2001}]{HajekTulipani2001}
H{\'a}jek, P., and Tulipani, S.
\newblock 2001.
\newblock Complexity of fuzzy probability logics.
\newblock {\em Fundamenta Informaticae} 45(3):207--213.

\bibitem[\protect\citeauthoryear{H\'{a}jek}{1998}]{Hajek1998}
H\'{a}jek, P.
\newblock 1998.
\newblock {\em {Metamathematics of Fuzzy Logic}}, volume~4 of {\em Trends in Logic}.
\newblock Dordrecht: Springer.

\bibitem[\protect\citeauthoryear{Hanikov{\'a}, Many{\`a}, and Vidal}{2023}]{HanikovaManyaVidal2023}
Hanikov{\'a}, Z.; Many{\`a}, F.; and Vidal, A.
\newblock 2023.
\newblock {The MaxSAT problem in the real-valued MV-algebra}.
\newblock In Ramanayake, R., and Urban, J., eds., {\em Automated Reasoning with Analytic Tableaux and Related Methods. TABLEAUX 2023}, volume 14278 of {\em Lecture Notes in Artificial Intelligence}. Cham: Springer Nature Switzerland.
\newblock  386--404.

\bibitem[\protect\citeauthoryear{Hanikov\'{a}}{2011}]{Hanikova2011MFL2}
Hanikov\'{a}, Z.
\newblock 2011.
\newblock {Computational Complexity of Propositional Fuzzy Logics}.
\newblock In Cintula, P.; H{\'{a}}jek, P.; and Noguera, C., eds., {\em {Handbook of Mathematical Fuzzy Logic}}, volume~38 of {\em Studies in logic}. College Publications.
\newblock  793--852.

\bibitem[\protect\citeauthoryear{Inoue and Kozhemiachenko}{2025}]{InoueKozhemiachenko2025arxiv}
Inoue, K., and Kozhemiachenko, D.
\newblock 2025.
\newblock Complexity of abduction in Łukasiewicz logic.
\newblock Available at: \href{https://arxiv.org/abs/2507.13847}{arXiv: 2507.13847}.

\bibitem[\protect\citeauthoryear{Josephson and Josephson}{2009}]{JosephsonJosephson2009}
Josephson, J., and Josephson, S., eds.
\newblock 2009.
\newblock {\em Abductive inference}.
\newblock Cambridge, England: Cambridge University Press.

\bibitem[\protect\citeauthoryear{Kakas, Michael, and Mourlas}{2000}]{KakasMichaelMourlas2000}
Kakas, A.; Michael, A.; and Mourlas, C.
\newblock 2000.
\newblock {ACLP: Abductive constraint logic programming}.
\newblock {\em The Journal of Logic Programming} 44(1-3):129--177.

\bibitem[\protect\citeauthoryear{Koitz-Hristov and Wotawa}{2018}]{Koitz-HristovWotawa2018}
Koitz-Hristov, R., and Wotawa, F.
\newblock 2018.
\newblock Applying algorithm selection to abductive diagnostic reasoning.
\newblock {\em Applied Intelligence} 48(11):3976--3994.

\bibitem[\protect\citeauthoryear{Krieken, Acar, and Harmelen}{2022}]{vanKriekenAcarvanHarmelen2022}
Krieken, E.~v.; Acar, E.; and Harmelen, F.~v.
\newblock 2022.
\newblock Analyzing differentiable fuzzy logic operators.
\newblock {\em Artificial Intelligence} 302(103602):103602.

\bibitem[\protect\citeauthoryear{Liang \bgroup et al\mbox.\egroup }{2022}]{LiangWangZhouYang2022}
Liang, C.; Wang, W.; Zhou, T.; and Yang, Y.
\newblock 2022.
\newblock {Visual Abductive Reasoning}.
\newblock In {\em 2022 {IEEE/CVF} Conference on Computer Vision and Pattern Recognition ({CVPR})}.
\newblock IEEE.

\bibitem[\protect\citeauthoryear{Magnani}{2011}]{Magnani2001}
Magnani, L.
\newblock 2011.
\newblock {\em Abduction, reason and science}.
\newblock New York, NY: Springer, 2001 edition.

\bibitem[\protect\citeauthoryear{Marquis}{1995}]{Marquis1995}
Marquis, P.
\newblock 1995.
\newblock {Knowledge Compilation Using Theory Prime Implicates}.
\newblock In {\em Proceedings of the Fourteenth International Joint Conference on Artificial Intelligence (IJCAI-95)}, volume~I,  837--845.

\bibitem[\protect\citeauthoryear{Marquis}{2000}]{Marquis2000HDRUMS}
Marquis, P.
\newblock 2000.
\newblock {Consequence Finding Algorithms}.
\newblock In {\em Handbook of Defeasible Reasoning and Uncertainty Management Systems: Algorithms for Uncertainty and Defeasible Reasoning}. Springer.
\newblock  41--145.

\bibitem[\protect\citeauthoryear{McNaughton}{1951}]{McNaughton1951}
McNaughton, R.
\newblock 1951.
\newblock A theorem about infinite-valued sentential logic.
\newblock {\em The Journal of Symbolic Logic} 16(1):1--13.

\bibitem[\protect\citeauthoryear{Mellouli and Bouchon-Meunier}{2003}]{MellouliBouchon-Meunier2003}
Mellouli, N., and Bouchon-Meunier, B.
\newblock 2003.
\newblock Abductive reasoning and measures of similitude in the presence of fuzzy rules.
\newblock {\em Fuzzy Sets And Systems} 137(1):177--188.

\bibitem[\protect\citeauthoryear{Miyata, Furuhashi, and Uchikawa}{1998}]{MiyataFuruhashiUchikawa1998}
Miyata, Y.; Furuhashi, T.; and Uchikawa, Y.
\newblock 1998.
\newblock A proposal of abductive inference with degrees of manifestations.
\newblock {\em International Journal of Intelligent Systems} 13(6):467--481.

\bibitem[\protect\citeauthoryear{Mundici}{1987}]{Mundici1987}
Mundici, D.
\newblock 1987.
\newblock Satisfiability in many-valued sentential logic is {NP}-complete.
\newblock {\em Theoretical Computer Science} 52(1):145--153.

\bibitem[\protect\citeauthoryear{Paul}{1993}]{Paul1993}
Paul, G.
\newblock 1993.
\newblock Approaches to abductive reasoning: an overview.
\newblock {\em Artificial Intelligence Review} 7(2):109--152.

\bibitem[\protect\citeauthoryear{Pavelka}{1979a}]{Pavelka1979FL1}
Pavelka, J.
\newblock 1979a.
\newblock {On Fuzzy Logic I. Many‐valued rules of inference}.
\newblock {\em Mathematical Logic Quarterly} 25(3-6):45--52.

\bibitem[\protect\citeauthoryear{Pavelka}{1979b}]{Pavelka1979FL2}
Pavelka, J.
\newblock 1979b.
\newblock {On Fuzzy Logic II. Enriched residuated lattices and semantics of propositional calculi}.
\newblock {\em Mathematical Logic Quarterly} 25(7-12):119--134.

\bibitem[\protect\citeauthoryear{Pavelka}{1979c}]{Pavelka1979FL3}
Pavelka, J.
\newblock 1979c.
\newblock {On Fuzzy Logic III. Semantical completeness of some many‐valued propositional calculi}.
\newblock {\em Mathematical Logic Quarterly} 25(25-29):447--464.

\bibitem[\protect\citeauthoryear{Pfandler, Pichler, and Woltran}{2015}]{PfandlerPichlerWoltran2015}
Pfandler, A.; Pichler, R.; and Woltran, S.
\newblock 2015.
\newblock {The Complexity of Handling Minimal Solutions in Logic-Based Abduction}.
\newblock {\em Journal of Logic and Computation} 25(3):805--825.

\bibitem[\protect\citeauthoryear{Pichler and Woltran}{2010}]{PichlerWoltran2010}
Pichler, R., and Woltran, S.
\newblock 2010.
\newblock {The Complexity of Handling minimal Solutions in Logic-Based Abduction}.
\newblock In {\em Proceedings of 19th European Conference on Artificial Intelligence (ECAI-2010)}, volume 215 of {\em Frontiers in Artificial Intelligence and Applications}. IOS Press.
\newblock  895--900.

\bibitem[\protect\citeauthoryear{Poole}{1989}]{Poole1989}
Poole, D.
\newblock 1989.
\newblock {Explanation and Prediction: An Architecture for Default and Abductive Reasoning}.
\newblock {\em Computational Intelligence} 5(2):97--110.

\bibitem[\protect\citeauthoryear{Poole}{1993}]{Poole1993}
Poole, D.
\newblock 1993.
\newblock {Probabilistic Horn abduction and Bayesian networks}.
\newblock {\em Artificial Intelligence} 64(1):81--129.

\bibitem[\protect\citeauthoryear{Rodriguez \bgroup et al\mbox.\egroup }{2022}]{RodriguezTuytEstevaGodo2022}
Rodriguez, R.; Tuyt, O.; Esteva, F.; and Godo, L.
\newblock 2022.
\newblock {Simplified Kripke semantics for K45-like G{\"o}del modal logics and its axiomatic extensions}.
\newblock {\em Studia Logica} 110(4):1081--1114.

\bibitem[\protect\citeauthoryear{Sakama and Inoue}{1995}]{SakamaInoue1995}
Sakama, C., and Inoue, K.
\newblock 1995.
\newblock {The Effect of Partial Deduction in Abductive Reasoning}.
\newblock In {\em {Logic Programming: The 12th International Conference}}. The MIT Press.

\bibitem[\protect\citeauthoryear{Sato, Ishihata, and Inoue}{2011}]{SatoIshihataInoue2011}
Sato, T.; Ishihata, M.; and Inoue, K.
\newblock 2011.
\newblock Constraint-based probabilistic modeling for statistical abduction.
\newblock {\em Machine Learning} 83(2):241--264.

\bibitem[\protect\citeauthoryear{Shanahan}{2005}]{Shanahan2005}
Shanahan, M.
\newblock 2005.
\newblock Perception as abduction: turning sensor data into meaningful representation.
\newblock {\em Cognitive Science} 29(1):103--134.

\bibitem[\protect\citeauthoryear{Stickel}{1990}]{Stickel1990}
Stickel, M.
\newblock 1990.
\newblock {Rationale and Methods for Abductive Reasoning in Natural-Language Interpretation}.
\newblock In {\em Natural Language and Logic}, volume 459 of {\em Lecture Notes in Artificial Intelligence},  233--252.
\newblock Springer.

\bibitem[\protect\citeauthoryear{Stoian, Giunchiglia, and Lukasiewicz}{2023}]{StoianGiunchigliaLukasiewicz2023}
Stoian, M.; Giunchiglia, E.; and Lukasiewicz, T.
\newblock 2023.
\newblock {Exploiting T-norms for Deep Learning in Autonomous Driving}.
\newblock In Avila~Garcez, A.; Besold, T.; Gori, M.; and Jiménez-Ruiz, E., eds., {\em Proceedings of the 17th International Workshop on Neural-Symbolic Learning and Reasoning, NeSy 2023},  369--380.

\bibitem[\protect\citeauthoryear{Straccia}{2016}]{Straccia2016}
Straccia, U.
\newblock 2016.
\newblock {\em Foundations of fuzzy logic and semantic web languages}.
\newblock Chapman \& Hall/CRC Studies in Informatics Series. Philadelphia, PA: Chapman \& Hall/CRC.

\bibitem[\protect\citeauthoryear{Tsypyschev}{2017}]{Tsypyschev2017}
Tsypyschev, V.
\newblock 2017.
\newblock Application of risk theory approach to fuzzy abduction.
\newblock In {\em Cybernetics and Mathematics Applications in Intelligent Systems}, volume 574 of {\em Advances in Intelligent Systems and Computing}. Cham: Springer International Publishing.
\newblock  13--19.

\bibitem[\protect\citeauthoryear{Vojt{\'a}{\v s}}{2001}]{Vojtas2001}
Vojt{\'a}{\v s}, P.
\newblock 2001.
\newblock Fuzzy logic programming.
\newblock {\em Fuzzy Sets And Systems} 124(3):361--370.

\bibitem[\protect\citeauthoryear{Vojt{\'{a}}\v{s}}{1999}]{Vojtas1999}
Vojt{\'{a}}\v{s}, P.
\newblock 1999.
\newblock {Fuzzy Logic Abduction}.
\newblock In Mayor, G., and Su{\~{n}}er, J., eds., {\em Proceedings of the {EUSFLAT-ESTYLF} Joint Conference},  319--322.
\newblock Universitat de les Illes Balears, Palma de Mallorca, Spain.

\bibitem[\protect\citeauthoryear{Yamada and Mukaidono}{1995}]{YamadaMukaidono1995}
Yamada, K., and Mukaidono, M.
\newblock 1995.
\newblock {Fuzzy Abduction Based on \L{}ukasiewicz Infinite-Valued Logic and Its Approximate Solutions}.
\newblock In {\em Proceedings of 1995 IEEE International Conference on Fuzzy Systems.}, volume~1,  343--350.
\newblock IEEE.

\bibitem[\protect\citeauthoryear{Zadeh}{1965}]{Zadeh1965}
Zadeh, L.
\newblock 1965.
\newblock Fuzzy sets.
\newblock {\em Information and Control} 8(3):338--353.

\bibitem[\protect\citeauthoryear{Zadeh}{1975}]{Zadeh1975}
Zadeh, L.
\newblock 1975.
\newblock Fuzzy logic and approximate reasoning.
\newblock {\em Synthese} 30(3-4):407--428.

\end{thebibliography}
% \end{document}
\clearpage
\appendix
\section{Proofs of Section~\ref{sec:intervalterms}}
\intervaltermcomplexity*
\begin{proof}
We begin with Item~1. For an interval literal $\lambda$ over $p\in\Prop$, define
\begin{align*}
\interval(p,p\leq\cmbf)&=[0,c]&\interval(p,p\geq\cmbf)&=[c,1]\\
\interval(p,p<\cmbf)&=[0,c)&\interval(p,p>\cmbf)&=(c,1]
\end{align*}
It is clear that $v(\lambda)=1$ iff $v(p)\in\interval(p,\lambda)$ (and $v(\lambda)=0$, otherwise). Now for an interval term $\sigma$ and $p\in\Prop(\sigma)$, set
\begin{align*}
\interval(p,\sigma)&=\bigcap\limits_{\lambda\in\sigma}\interval(p,\lambda)
\end{align*}
One can straightforwardly verify that $\sigma$ is $\Luk$-satisfiable iff $\interval(p,\sigma)\neq\varnothing$ for every $p\in\Prop(\sigma)$. Moreover, if $\sigma$ is not satisfiable, then $v(\sigma)=0$ for every $v$.

Now let $\sigma$ and $\tau$ be two \emph{$\Luk$-satisfiable} interval terms. We show that $\sigma\models_\Luk\tau$ iff the following two conditions are met:
\begin{enumerate}[$(a)$,noitemsep,topsep=1pt]
\item $\Prop(\sigma)\supseteq\Prop(\tau)$;
\item $\interval(p,\tau)\subseteq\interval(p,\sigma)$ for every $p\in\Prop(\sigma)$.
\end{enumerate}

Assume that $(a)$ and $(b)$ hold and that $v(\sigma)=1$. Then, clearly $v(p)\in\interval(p,\sigma)$ for every $p\in\Prop(\sigma)$. By $(a)$ and~$(b)$, we have $v(p)\in\interval(p,\tau)$ for every $p\in\Prop(\tau)$, whence $v(\tau)=1$.

Conversely, assume that $(a)$ does not hold and that there is some $p\in\Prop(\tau)\setminus\Prop(\sigma)$. Since $\sigma$ is satisfiable, let $v(\sigma)=1$ and set $v(p)\notin\interval(p,\tau)$. It is clear that $v(\tau)=0$. Now, assume that $(b)$ does not hold, i.e., there is some $p\in\Prop(\sigma)$ s.t.\ $\interval(p,\tau)\not\subseteq\interval(p,\sigma)$. Now let $v$ be an $\Luk$-valuation s.t.\ $v(\sigma)=1$ and $v(p)=x$ for some $x\in\interval(p,\sigma)\setminus\interval(p,\tau)$. It is clear that $v(\tau)=0$.

Finally, given two interval terms $\sigma$ and $\tau$, one of the three options holds: either $\sigma$ is unsatisfiable (whence, $\sigma\models_\Luk\tau$), $\sigma$ is satisfiable but $\tau$ is not (whence, $\sigma\not\models_\Luk\tau$), or both are satisfiable. As it takes polynomial time to verify whether there is some $p\in\Prop(\sigma)$ or $q\in\Prop(\tau)$ s.t.\ $\interval(p,\sigma)=\varnothing$ or $\interval(q,\tau)=\varnothing$, the satisfiability of interval terms can be checked in polynomial time. Finally, for the third case, we need to verify that $(a)$ and $(b)$ hold which can also be done in polynomial time.

For Item~2, the membership is evident from Proposition~\ref{prop:Lukintnpcomplete} and hardness can be obtained by the following reduction from $\Luk$-validity. Namely, $\chi\in\LLukint$ is $\Luk$-valid (i.e., is entailed by the empty set of formulas) iff $p\geq\mathbf{1}\models_\Luk\chi$ for some $p\notin\Prop(\chi)$.

The hardness part of Item~3 is obtained by the reduction from $\Luk$-unsatisfiability. Let $\Gamma=\{\phi\}$. It is clear that $\phi$ is $\Luk$-un\-sa\-tis\-fiable iff $\phi,p\geq\mathbf{1}\models q\geq\mathbf{1}$ for some $p,q\notin\Prop(\phi)$. The membership follows from Proposition~\ref{prop:Lukintnpcomplete}.
\end{proof}
\section{Proofs of Section~\ref{sec:Lukabduction}}
\CPLtoLuk*
\begin{proof}
Observe that given an $\Luk$-valuation $v$, $v(p\vee\neg p)=1$ iff $v(p)\in\{0,1\}$. Furthermore, on $\{0,1\}$ $\odot$ works the same as $\wedge$ and $\oplus$ as $\vee$ (cf.~Definition~\ref{def:Luksemantics} and~\eqref{equ:weakconnectives}). Now let $v$ be a~\emph{classical valuation} witnessing $\Gamma\not\models_\CPL\chi$. It is clear that $v$~also witnesses $\Gamma^\Luk,\{p\vee\neg p\mid p\in\Prop[\Gamma\cup\{\chi\}]\}\not\models_\Luk\chi$. Conversely, let $v$ be an $\Luk$-valuation witnessing $\Gamma^\Luk,\{p\vee\neg p\mid p\in\Prop[\Gamma\cup\{\chi\}]\}\not\models_\Luk\chi$. As we have just observed, this means that all variables are evaluated over $\{0,1\}$. As $\Gamma^\Luk$ and $\chi^\Luk$ are obtained from $\Gamma$ and~$\chi$ by replacing $\wedge$ and $\vee$ with $\odot$ and $\oplus$, it follows that $\Gamma\not\models_\CPL\chi$.
\end{proof}
% \subsection{Proof of Theorem~\ref{theorem:Lukminimalsolutionrecognitioncomplexity}}
\begin{definition}\label{def:primeimplicant}
Let $\chi\in\LCPL$ and $\tau$ be a~weak conjunction of simple literals. We say that $\tau$ is a~\emph{prime implicant of}~$\chi$ iff $\tau\models_\CPL\phi$ and there is no other weak conjunction $\tau'$ of simple literals s.t.\ $\tau'\models_\CPL\chi$, $\tau\models_\CPL\tau'$, and $\tau'\not\models_\CPL\tau$.
\end{definition}

The next definition introduces ‘next weakest literals’ given $p\lozenge\cmbf$ and $\Hmsf$. The idea is to incrementally extend the interval corresponding to $p\lozenge\cmbf$.
\begin{definition}\label{def:nextweakestliteral}
For a~finite set of interval literals $\Hmsf$, $\lambda\in\Hmsf$, and $\cmbf\in[0,1]_\Qmbb$, we define
\begin{align*}
% \cmbf^\downarrow(p)&=\max\left\{\left\{\cmbf'\!\in\!\Qmbb\left|\begin{matrix}\cmbf'\!<\!\cmbf\\p\!\geq\!\cmbf'\!\in\!\Hmsf\end{matrix}\right.\right\}\!\cup\!\left\{\cmbf'\!\in\!\Qmbb\left|\begin{matrix}\cmbf'\!<\!\cmbf\\p\!>\!\cmbf'\!\in\!\Hmsf\end{matrix}\right.\right\}\right\}\\
% \cmbf^\uparrow(p)&=\min\left\{\left\{\cmbf'\!\in\!\Qmbb\left|\begin{matrix}\cmbf'\!>\!\cmbf\\p\!\leq\!\cmbf'\!\in\!\Hmsf\end{matrix}\right.\right\}\!\cup\!\left\{\cmbf'\!\in\!\Qmbb\left|\begin{matrix}\cmbf'\!>\!\cmbf\\p\!<\!\cmbf'\!\in\!\Hmsf\end{matrix}\right.\right\}\right\}\\
\cmbf^\uparrow(p)&=\min\left\{\cmbf'\!\in\!\Qmbb\mid \cmbf'\!>\!\cmbf\text{ and }p\!\leq\!\cmbf'\in\Hmsf;\text{ or }p\!<\!\mathbf{c}'\in\Hmsf\right\}\\
\cmbf^\downarrow(p)&=\max\left\{\cmbf'\!\in\!\Qmbb\mid \cmbf'\!<\!\cmbf\text{ and }p\!\geq\!\cmbf'\in\Hmsf;\text{ or }p\!>\!\mathbf{c}'\!\in\!\Hmsf\right\}
%{ c' ∈ ℚ | c' < c and p ≥ c' ∈ H or p > c' ∈ H }
\end{align*}

Now, let $\cmbf^\downarrow(p)$ and $\dmbf^\uparrow(p)$ exist in $\Hmsf$ for~$p$. We define:
\begin{align*}
(p\geq\cmbf)^\flat_\Hmsf&=
\begin{cases}
p\geq\cmbf^\downarrow&\text{if }p\geq\cmbf^\downarrow\in\Hmsf\\
p>\cmbf^\downarrow&\text{otherwise}
\end{cases}\\
(p>\cmbf)^\flat_\Hmsf&=
\begin{cases}
p\geq\cmbf&\text{if }p\geq\cmbf\in\Hmsf\\
p\geq\cmbf^\downarrow&\text{if }p\geq\cmbf^\downarrow\in\Hmsf\\
p>\cmbf^\downarrow&\text{otherwise}
\end{cases}
\end{align*}
\begin{align*}
(p\leq\dmbf)^\flat_\Hmsf&=
\begin{cases}
p\leq\dmbf^\uparrow&\text{if }p\geq\dmbf^\downarrow\in\Hmsf\\
p>\dmbf^\uparrow&\text{otherwise}
\end{cases}\\
(p<\dmbf)^\flat_\Hmsf&=
\begin{cases}
p\leq\dmbf&\text{if }p\leq\dmbf\in\Hmsf\\
p\leq\dmbf^\uparrow&\text{if }p\leq\dmbf^\downarrow\in\Hmsf\\
p<\dmbf^\uparrow&\text{otherwise}
\end{cases}
\end{align*}
\end{definition}
\Lukminimalsolutionrecognitioncomplexity*
\begin{proof}
For the hardness, we provide a~reduction from the prime implicant recognition in classical logic which is known to be $\DP$-complete~\cite[Proposition~111]{Marquis2000HDRUMS}. Let w.l.o.g.\ $\chi$ be classically satisfiable and $\tau$ a~prime implicant of $\chi$. Thus, $\tau\consvDashCPL\chi$ and there is no weak conjunction of simple literals $\sigma$ s.t.\ $\sigma\consvDashCPL\chi$, $\tau\models_\CPL\sigma$, and $\sigma\not\models_\CPL\tau$. We show that $\tau^\odot$ is an $\Luk$-minimal solution to $\Pmbb\!=\!\left\langle\Gamma,\chi^\Luk\!\odot\!q,\Hmsf\right\rangle$ with $q\!\notin\!\Prop(\chi)$ and
\begin{align*}
\Gamma&=\left\{p\vee\neg p\mid p\in\Prop(\chi\wedge q)\right\}\cup\{q\}\\
\Hmsf&=\{p\geq\mathbf{1}\mid p\in\Prop(\chi)\}\cup\{p\leq\mathbf{0}\mid p\in\Prop(\chi)\}
\end{align*}

First, it is clear that all solutions of $\Pmbb$ are proper because $q$ does not occur in $\Hmsf$. Now, consider~$\tau^\odot$. It is clear that $\Gamma,\tau^\odot\consvDashLuk\chi^\Luk\odot q$. Indeed, $\tau$ was classically satisfiable, whence $\tau^\odot$ is $\Luk$-satisfiable by the same valuation $v$ that satisfied $\tau$ since $\odot$ and $\wedge$ behave identically on $\{0,1\}$. But it also means that $\Gamma\cup\{\tau\}$ is $\Luk$-satisfiable as $v(p\vee\neg p)=1$ iff $v(p)\in\{0,1\}$ and we can set $v(q)=1$. Furthermore, if for some $\Luk$-valuation $v'$, $v'(\tau^\odot)=1$ and $v'(\phi)=1$ for every $\phi\in\Gamma$, it means that all variables are evaluated over $\{0,1\}$, i.e., $v'$ is a~classical valuation. Hence, $v'(\chi^\odot\odot q)=1$, as required. It remains to see that there is no solution $\sigma$ s.t.\ $\tau^\odot\models_\Luk\sigma$ and $\sigma\not\models_\Luk\tau^\odot$. Assume for the sake of contradiction that there is such a~solution $\sigma$ and define
\begin{align*}
\sigma^\cl&=\bigwedge\limits_{p\!\geq\!\mathbf{1}\in\sigma}\!\!\!\!p\wedge\bigwedge\limits_{r\!\leq\!\mathbf{0}\in\sigma}\!\!\!\!\neg r
\end{align*}
We will show that $\sigma^\cl\!\!\not\models_\CPL\!\tau$, $\sigma^\cl\!\!\models_\CPL\!\chi$, and $\tau\!\!\models_\CPL\!\sigma^\cl$. Observe that $(\sigma^\cl)^\odot=\sigma$. Now, let $v$ be an $\Luk$-valuation witnessing $\sigma\not\models_\Luk\tau^\odot$. Since $\sigma$ and $\tau^\odot$ are interval terms containing only literals of the form $r\leq\mathbf{0}$ and $r\geq\mathbf{1}$, it means that $v(\sigma)=1$ and $v(\tau^\odot)=0$. Now, we define a~classical valuation $v^\cl$ as follows:
\begin{align*}
v^\cl(r)&=
\begin{cases}
1&\text{if }v(s\geq\mathbf{1})=1\text{ or }v(r\leq\mathbf{0})=0\\
0&\text{if }v(r\leq\mathbf{0})=1\text{ or }v(r\geq\mathbf{1})=0
\end{cases}
\end{align*}
It is clear that $v^\cl$ witnesses $\sigma^\cl\not\models_\CPL\tau$. Finally, using the definition of $\Gamma$, that $\sigma$ contains only literals of the form $r\leq\mathbf{0}$ and $r\geq\mathbf{1}$, and that all $\LLuk$ connectives behave classically on $\{0,1\}$, we obtain that $\Gamma,\sigma\consvDashLuk\chi\odot q$ entails $\sigma^\cl\models_\CPL\chi$ and $\tau^\odot\models_\Luk\sigma$ entails $\tau\models_\CPL\sigma^\cl$. This, however, means that $\tau$ is not a~prime implicant of $\chi$, contrary to the assumption.

For the converse, let $\tau\in\LukminSol(\Pmbb)$. Define $\tau_\cl=\bigwedge\limits_{(r\!\geq\!\mathbf{1})\in\tau}\!\!\!\!r\wedge\bigwedge\limits_{(s\!\leq\!\mathbf{0})\in\tau}\!\!\!\!\neg s$. From Proposition~\ref{prop:CPLtoLuk} and the fact that $\tau\simeq_\Luk\bigodot\limits_{(r\!\geq\!\mathbf{1})\in\tau}\!\!\!\!r\odot\bigodot\limits_{(s\!\leq\!\mathbf{0})\in\tau}\!\!\!\!\neg s$, we have $\tau_\cl\consvDashLuk\chi$. Now assume for contradiction that there is some $\sigma$ s.t.\ $\sigma\models_\CPL\chi$, $\tau\models_\CPL\sigma$, but $\sigma\not\models_\CPL\tau$. It is easy to check that in this case, $\sigma^\odot\in\Sol(\Pmbb)$, $\tau\models_\Luk\sigma^\odot$, and $\sigma^\odot\not\models_\Luk\tau$, i.e., $\tau$ is not a~$\models_\Luk$-minimal solution, contrary to the assumption.

Now, consider the membership. Given $\Pmbb=\langle\Gamma,\chi,\Hmsf\rangle$, $\tau=\bigodot\limits^n_{i=1}\lambda_i$, and $k\in\{1,\ldots,n\}$, let
\begin{align*}
\tau^\flat_{\lambda_k}&=
\begin{cases}
% ⨀_{λ' ∈ τ, λ' ≠ λ} λ'
\bigodot\limits_{\scriptsize{\begin{matrix}\lambda'\in\tau\\\lambda'\neq\lambda\end{matrix}}}\lambda'\odot{\lambda^\flat_k}_\Hmsf&\text{if }{\lambda^\flat_k}_\Hmsf\text{ is defined}\\
\bigodot\limits^{k-1}_{i=1}\lambda_i\odot\bigodot\limits^n_{j=k+1}\!\!\!\!\lambda_j&\text{otherwise}
\end{cases}
\end{align*}
We use an $\np$-oracle to guess $\Omc(|\Hmsf|)$ valuations $v_\mathsf{sat}$, $v_\mathsf{prp}$, and $v_\lambda$ (for each $\lambda\in\tau$) and verify whether they witness (i) $\Gamma,\tau\not\models_\Luk\bot$, (ii) $\tau\not\models_\Luk\chi$, and (iii) $\Gamma,\tau^\flat_{\lambda}\not\models_\Luk\chi$, respectively. Parallel to that (as we do not need the results of (i)--(iii)), we conduct a~$\conp$ check that (iv) $\Gamma,\tau\models_\Luk\chi$. It follows from the definition of $\models_\Luk$-minimal solutions that $\tau$ is a~$\models_\Luk$-minimal solution iff the $\np$ and $\conp$ checks succeed.
\end{proof}
\theoryminimalrecognition*
\begin{proof}
Let $\Pmbb=\langle\Gamma,\chi,\Hmsf\rangle$ be a~$\Luk$-abduction problem, and $\sigma$~is the interval term we are testing. We sketch a~$\Sigma^\Pmsf_2$ procedure for the complementary problem of testing whether $\sigma$~is \emph{not} a~theory-minimal solution. We first guess either (1) ‘not a proper solution’ or (2) another interval term $\sigma'$. In the case~(1), due to the $\DP$-completeness of proper solution recognition (Theorem~\ref{theorem:propersolutionrecognitioncomplexity}), we make one or two calls to an $\np$ oracle to verify that $\sigma$ is indeed not a~proper solution (in which case we return ‘yes’). In the case (2), we can make four calls to an $\np$ oracle: two to verify to verify that $\sigma'$ is a proper solution; one to verify that that $\Gamma,\sigma\models_\Luk\sigma'$, and another one to check that $\Gamma,\sigma'\not\models_\Luk\sigma$ (returning ‘yes’ if the calls show this to be the case). By Proposition~\ref{prop:intervaltermcomplexity}, these entailments can indeed be verified by an $\np$-oracle. The procedure we have just described will return yes iff the input $\sigma$ is not a~theory-minimal solution, which yields the desired membership in $\Pi^\Pmsf_2$ for the original problem. 
\end{proof}
\theoryrelevancerecognition*
\begin{proof}
We consider the relevance. Recall from Theorem~\ref{theorem:theoryminimalrecognition} that a~theory-minimal solution can be recognised in~$\Pi^\Pmsf_2$. Thus, it suffices to guess a~theory-minimal solution $\tau$ and then check that $\lambda$ indeed occurs there. It follows that the relevance of a~hypothesis w.r.t.\ the set of theory-minimal solutions can be decided in~$\Sigma^\Pmsf_3$. In the same way, we can show that \emph{non-necessity} of a~hypothesis can also be decided in~$\Sigma^\Pmsf_3$. Hence, necessity is decidable in~$\Pi^\Pmsf_3$.
\end{proof}
\section{Proofs of Section~\ref{sec:simpleclauseabduction}}
\begin{definition}\label{def:SCtolinearinequalities}
Let $\kappa=\bigoplus\limits^m_{i=1}p_i\oplus\bigoplus\limits^n_{j=1}\neg q_j$ be a~simple clause and $\lambda=p\lozenge\cmbf$ an interval literal. We define linear inequalities $\kappa^\lineq$ and $\lambda^\lineq$ as follows:
\begin{align*}
\kappa^\lineq&=\sum\limits^m_{i=1}x_{p_i}+\sum\limits^n_{j=1}(1-x_{q_j})\geq1&\lambda^\lineq&=x_{p}\lozenge c
\end{align*}

Now let
\begin{align*}
\Gamma&=\{\kappa_i\mid i\in\{1,\ldots,m\}\}\cup\{\tau_j\mid j\in\{1,\ldots,n\}\}
\end{align*}
be a~set containing simple clauses and interval terms s.t.\ $\tau_j=\bigodot\limits^{k_j}_{i=1}\lambda^j_i$. We define the system of linear inequalities $\Gamma^\lineq$ that corresponds to $\Gamma$ as follows:
\begin{align*}
\Gamma^\lineq&=\{\kappa^\lineq_i\!\mid i\!\in\!\{1,\ldots,m\}\}\!\cup\!\bigcup\limits^n_{j=1}\{(\lambda^j_i)^\lineq\!\mid i\!\in\!\{1,\ldots,k_j\}\}
\end{align*}
\end{definition}
\SCArecognitionpolynomial*
\begin{proof}
We provide the remaining details of the membership proof. We assume w.l.o.g.\ that $\sigma=\bigodot^n_{i=1}\lambda_i$ and that $\chi$ is an interval or a~simple clause. Indeed, it is easy to see that explaining a~simple or an interval term $\tau=\bigodot^m_{i=1}\lambda'_i$ is equivalent to explaining all its literals separately. To establish that $\sigma\in\Sol(\Pmbb)$, we need to check that (i) $\Gamma,\sigma\not\models_\Luk\bot$ and (ii) $\Gamma,\sigma\models_\Luk\chi$. For (i), observe from Definitions~\ref{def:Luksemantics} and~\ref{def:SCtolinearinequalities} that $\Gamma\cup\{\sigma\}$ is satisfiable iff $\Smsf=\Gamma^\lineq\cup\bigcup^n_{i=1}\{\lambda_i^\lineq\}$ has a~solution over $[0,1]$. For~(ii), we have two cases depending on whether $\chi$ is a~simple clause or an interval clause.

First, let $\chi=\bigoplus^{k}_{i=1}p_i\oplus\bigoplus^m_{i=1}\neg q_i$. We have that $\Gamma,\sigma\models_\Luk\chi$ iff $\Smsf\cup\left\{\sum^k_{i=1}x_{p_i}+\sum^m_{i=1}(1-x_{q_i})<1\right\}$ \emph{does not have} a~solution over $[0,1]$. Second, if $\chi=\bigoplus^k_{i=1}\!\lambda'_i$, then \mbox{$\Gamma,\sigma\models_\Luk\chi$} iff $\Smsf\cup\{\neg\lambda'^\lineq_i\mid1\!\leq\!i\!\leq\!k\}$ does not have a~solution over $[0,1]$.

As both (i) and (ii) can be determined in polynomial time, arbitrary solutions can be recognised in polynomial time too. For proper solutions, we need to check that $\sigma\not\models_\Luk\chi$. Again, it is easy to see that this check can also be done in polynomial time. Finally, as entailment between interval terms is decidable in polynomial time (Proposition~\ref{prop:intervaltermcomplexity}) and as there are only polynomially many candidates for every given interval term (recall the proof of Theorem~\ref{theorem:Lukminimalsolutionrecognitioncomplexity}), $\models_\Luk$-minimal solutions are recognisable in polynomial time as well.
\end{proof}
\SCAtheoryminimalrecognitionconp*
\begin{proof}
We consider the complementary problem: given $\tau$ and $\Pmbb=\langle\Gamma,\chi,\Hmsf\rangle$, verify that it is \emph{not} a~theory-minimal solution and show that it belongs to~$\np$. Indeed, it suffices to guess $\tau'$ s.t.\ $\tau'$ is a~proper solution to $\Pmbb$, $\Gamma,\tau\models_\Luk\tau'$, and $\Gamma,\tau'\not\models_\Luk\tau$. Theorem~\ref{theorem:SCArecognitionpolynomial} shows that these checks can be done in polynomial time.
\end{proof}
\intervalclauseabductioncomplexityrecognition*
\begin{proof}
Membership follows from Theorem~\ref{theorem:arbitrarysolutionrecognitioncomplexity}. We show the hardness via a~reduction from the classical solution recognition. Assume that $\Pmbb=\langle\Gamma,p,\Hmsf\rangle$ is a~\emph{classical abduction problem} with $\Gamma=\{\kappa_i\mid i\in\{1,\ldots,n\}\}$ being a~set of disjunctive clauses and that $\tau$ is a~solution to $\Pmbb$. We define $\Pmbb_\interval=\langle\Gamma_\interval,p,\Hmsf_\interval\rangle$ and $\tau_\interval$ as follows:
\begin{align}\label{equ:intervalclausereduction}
\Gamma_\interval&=\left\{\bigoplus\limits_{q_i\in\kappa_i}(q_i\!\geq\!\mathbf{1})\oplus\bigoplus\limits_{\neg r_i\in\kappa_i}(r_i\!<\!\mathbf{1})\mid i\in\{1,\ldots,n\}\right\}\nonumber\\
\Hmsf_\interval&=\{s\geq\mathbf{1}\mid s\in\Hmsf\}\cup\{t<\mathbf{1}\mid\neg t\in\Hmsf\}\nonumber\\
\tau_\interval&=\bigodot\limits_{s'\in\tau}(s'\geq\mathbf{1})\odot\bigodot\limits_{\neg t'\in\tau}(t<\mathbf{1})
\end{align}
Now let $\tau\!\in\!\Sol(\Pmbb)$, i.e., $\Gamma,\tau\!\consvDashCPL\!p$. Clearly, $\Gamma_\interval,\tau_\interval\!\not\models_\Luk\!\bot$. For the entailment, assume for contradiction that $v$ witnesses $\Gamma_\interval,\tau_\interval\not\models_\Luk p$. Consider a classical valuation $v_\cl$ s.t. $v_\cl(q)=1$ iff $v(q)=1$. One can easily see that $v_\cl$ witnesses $\Gamma,\tau\not\models_\CPL p$. For the converse direction, let $\tau\in\Sol(\Pmbb_\interval)$. A~solution for $\Pmbb$ can be obtained by replacing $r\geq\mathbf{1}$ with $r$ and $r<1$ with $\neg r$. The reasoning is similar.
\end{proof}
\coverfreesolutionrecognition*
\begin{proof}
For the membership, we construct a~reduction to the classical \emph{Horn} abduction for which solution recognition is polynomial. Let $\Gamma=\{\kappa_1,\ldots,\kappa_m,\kappa'_1,\ldots\kappa'_{m'}\}$ be a~set of interval clauses of the form $\kappa_i=\bigodot_{j=1}^{n_i}(p_j\lozenge\cmbf_j)\rightarrow(q_i\lozenge\dmbf_i)$ and $\kappa'_i=\bigodot_{j=1}^{n'_i}(p'_j\lozenge\cmbf_j)\rightarrow\bot$.
% \begin{align*}
% \kappa_i&=\bigoplus\limits_{j=1}^{n_i}(p_j\lozenge\cmbf_j)
% \end{align*}
% for each $i\in\{1,\ldots,m\}$. Using Definition~\ref{def:Luksemantics} as well as~\eqref{equ:intervaltermsintervalclausesduality} and~\eqref{equ:alternativeclauses}, each clause $\kappa_i$ can be rewritten as follows:
% $\kappa^\odot_i=\bigodot_{j=1}^{n_i}(p_j\blacklozenge\cmbf_j)\rightarrow\bot$. Note that $\kappa_i\equiv_\Luk\kappa^\odot_i$.
% \begin{align*}
% \kappa^\odot_i&=\bigodot\limits_{j=1}^{n_i}(p_j\blacklozenge\cmbf_j)\rightarrow\bot
% \end{align*}
% Now set $\Gamma^\odot=\{\kappa^\odot_i\mid\kappa_i\in\Gamma\}$ and let $\Pmbb^\odot=\langle\Gamma^\odot,\tau,\Hmsf\rangle$. It is clear that the sets of (resp., proper, $\Luk$-minimal) solutions of $\Pmbb$ and $\Pmbb^\odot$ coincide.

Now, for each interval literal $\lambda$, let $r_\lambda$ be a~fresh propositional variable. We define 
\begin{align*}
\kappa_\Hmc&\coloneqq\bigwedge_{j=1}^{n_i}r_{p_j\lozenge\cmbf_j}\rightarrow r_{q_i\lozenge\dmbf_i}&
\kappa'_\Hmc&\coloneqq\bigwedge_{j=1}^{n'_i}r_{p'_j\lozenge\cmbf_j}\rightarrow\bot
\end{align*}

We consider the following \emph{classical} abduction problem: $\Pmbb_\Hmc=\langle\Gamma_\Hmc,\tau_\Hmc,\Hmsf_\Hmc\rangle$ defined as described below.
\begin{align}\label{equ:Hornreductionproblem}
\Gamma_\Hmc&=\left\{\kappa_\Hmc\mid\kappa\in\Gamma\right\}\cup\{\kappa'_\Hmc\mid\kappa'\in\Gamma\}\nonumber\\
% &\quad~~\left\{\bigwedge\limits^{n_i}_{j=1}r_{p_j\blacklozenge\cmbf_j}\rightarrow\bot\mid i\in\{1,\ldots,m\}\right\}\cup\nonumber\\
% &\quad~~\left\{p_{p_j\blacklozenge c_j}\!\wedge\!p_{p_k\blacklozenge c_k}\!\rightarrow\!\bot\mid\Luk\!\models\!\neg((p_j\blacklozenge c_j)\!\odot\!(p_k\blacklozenge c_k))\right\}\\
&\quad~\left\{r_\lambda\!\wedge\!r_{\lambda'}\!\rightarrow\!\bot\left|\begin{matrix}\Luk\!\models\!\lambda\odot\lambda'\rightarrow\bot\\\lambda\text{ and }\lambda'\text{ occur in }\Pmbb^\odot\end{matrix}\right.\right\}\cup\nonumber\\
&\quad~\left\{r_\lambda\rightarrow r_{\lambda'}\left|\begin{matrix}\Luk\models\lambda\rightarrow\lambda'\\\lambda\text{ and }\lambda'\text{ occur in }\Pmbb^\odot
\end{matrix}\right.\right\}\nonumber\\
\tau_\Hmc&=\bigwedge\limits_{\lambda\in\tau}r_\lambda\nonumber\\
\Hmsf_\Hmc&=\{r_\lambda\mid\lambda\in\Hmsf\}
\end{align}

As one can see, $\Gamma_\Hmc$ is polynomial in the size of $\Gamma$ and (recall Proposition~\ref{prop:intervaltermcomplexity}) it takes polynomial time to determine whether $\lambda\rightarrow\lambda'$ is $\Luk$-valid and $\lambda\odot\lambda'$ is $\Luk$-unsatisfieable. Furthermore, $\Gamma_\Hmc$ is a~set of Horn formulas and only contains variables of the form $r_\lambda$ with $\lambda$ being an interval term occurring in~$\Pmbb$. We show that $\sigma$~is a~(proper) solution to $\Pmbb$ iff $\sigma_\Hmc=\bigwedge\limits_{\lambda\in\sigma}r_\lambda$ is a~(proper) solution to $\Pmbb_\Hmc$.

Assume that $\sigma\!\in\!\Sol(\Pmbb)$. We show that $\Gamma_\Hmc,\sigma_\Hmc\!\consvDashCPL\!\tau_\Hmc$. Let $v$ be an $\Luk$-valuation witnessing that $\Gamma,\sigma\not\models_\Luk\bot$. We construct the classical valuation $v_\Hmc$ as follows:
\begin{align}
v_\Hmc(r_\lambda)&=
\begin{cases}
1&\text{if }v(\lambda)=1\\
0&\text{otherwise}
\end{cases}\label{equ:Hornvaluation}
\end{align}
It is clear from the construction of $\Gamma_\Hmc$ and $\sigma_\Hmc$ that $v_\Hmc$ witnesses $\Gamma_\Hmc,\sigma_\Hmc\not\models_\CPL\bot$.

Now let $\vmbf$ be a~classical valuation s.t.\ $\vmbf(\phi)=1$ for all $\phi\in\Gamma_\Hmc$, $\vmbf(\sigma_\Hmc)=1$ but $\vmbf(\tau_\Hmc)=0$. We need to define an $\Luk$-valuation $v_\vmbf$ s.t.\ $v_\vmbf(\lambda)=1$ iff $\vmbf(r_\lambda)=1$ would hold for every $\lambda$ occurring in~$\Pmbb$. For every $p_j\in\Prop[\Pmbb]$, we define 
\begin{align}
\Vmc_{p_j}&=\bigcap\limits_{\vmbf(r_{p_j\geq\cmbf_j})=1}\!\!\!\!\!\!\!\![c_j,1]\cap\!\!\!\!\!\bigcap\limits_{\vmbf(r_{p_j>\cmbf'_j})=1}\!\!\!\!\!\!\!\!(c'_j,1]\cap\nonumber\\
&\quad~\bigcap\limits_{\vmbf(r_{p_j\leq\cmbf''_j})=0}\!\!\!\!\!\!\!\!(c''_j,1]\cap\!\!\!\!\!\bigcap\limits_{\vmbf(r_{p_j<\cmbf'''_j})=0}\!\!\!\!\!\!\!\![c'''_j,1]\cap\nonumber\\
&\quad~\bigcap\limits_{\vmbf(r_{p_j\geq\dmbf_j})=0}\!\!\!\!\!\!\!\![0,d_j)\cap\!\!\!\!\!\bigcap\limits_{\vmbf(r_{p_j>\dmbf'_j})=0}\!\!\!\!\!\!\!\![0,d'_j]\cap\nonumber\\
&\quad~\bigcap\limits_{\vmbf(r_{p_j\leq\dmbf''_j})=1}\!\!\!\!\!\!\!\![0,d''_j]\cap\!\!\!\!\!\bigcap\limits_{\vmbf(r_{p_j<\dmbf'''_j})=1}\!\!\!\!\!\!\!\![0,d'''_j)\label{equ:valuesets}
\end{align}
It is clear that $\Vmc_{p_j}\neq\varnothing$ for each $p_j$. Indeed, otherwise, there would be two variables $r_\lambda$ and $r_{\lambda'}$ s.t.\ $\vmbf(r_\lambda)=\vmbf(r_{\lambda'})=1$ and $\lambda$ and $\lambda'$ are two interval literals over $p_j$ with $\lambda\odot\lambda'$ being $\Luk$-unsatisfiable. Hence, $r_\lambda\wedge r_{\lambda'}\rightarrow\bot\in\Gamma_\Hmc$, which means that $\vmbf$ does not satisfy $\Gamma_\Hmc$, contrary to the assumption.

We now set for every $p_j$, $v_\vmbf(p_j)=x$ for some $x\in\Vmc_{p_j}$. It is easy to see that $v_\vmbf(\theta)=1$ for all $\theta\!\in\!\Gamma$ and $v_\vmbf(\sigma)\!=\!1$ but $v_\vmbf(\tau)\!=\!0$. Thus, $\Gamma,\sigma\not\models_\Luk\tau$ which contradicts the assumption that $\sigma\!\in\!\Sol(\Pmbb)$. Thus, $\Gamma_\Hmc,\sigma_\Hmc\!\models_\CPL\!\tau_\Hmc$. It follows now that $\sigma_\Hmc$ is a~solution to~$\Pmbb_\Hmc$.

Conversely, let $\sigma_\Hmc$ be a~solution to~$\Pmbb_\Hmc$, that is, $\Gamma_\Hmc,\sigma_\Hmc\consvDashCPL\tau_\Hmc$. Assume that $\vmbf$ is a~classical valuation witnessing $\Gamma_\Hmc,\sigma_\Hmc\not\models_\CPL\bot$. We define $v_\vmbf$ using~\eqref{equ:valuesets} as in the previous part of the proof which gives us that $v_\vmbf$ witnesses $\Gamma,\sigma\not\models_\Luk\bot$. It remains to see that $\Gamma,\sigma\models_\Luk\tau$. To do that, assume for contradiction that there is some $\Luk$-valuation witnessing $\Gamma,\sigma\not\models_\Luk\tau$. We define $v_\Hmc$ as shown in~\eqref{equ:Hornvaluation}. It is clear that $v_\Hmc$ would witness $\Gamma_\Hmc,\sigma_\Hmc\not\models\tau_\Hmc$, contrary to the assumption. Hence, $\sigma\in\Sol(\Pmbb)$, as required.

Thus, we have that recognition of arbitrary solutions to problems whose theories are sets of cover-free clauses is polynomial. As our observation is also an interval term, it follows from Proposition~\ref{prop:intervaltermcomplexity} that verifying $\sigma\not\models_\Luk\tau$ takes polynomial time. This means that determining whether $\sigma$ is a~\emph{proper} solution also takes polynomial time. Finally, to see that recognition of \emph{entailment-minimal} solutions is polynomial as well, recall from the proof of Theorem~\ref{theorem:Lukminimalsolutionrecognitioncomplexity} that there are only polynomially many candidate interval terms we need to check for a~given~$\sigma$.

For the hardness, we construct a~logspace reduction of the solution recognition for classical Horn abduction problems to the solution recognition for Łukasiewicz abduction problems with cover-free theories. Namely, let $\Pmbb=\langle\Gamma,\tau,\Hmsf\rangle$ be a~classical Horn abduction problem with $\Gamma$ represented as Horn implications $\eta=\bigwedge^m_{i=1}p_i\rightarrow q_i$ and $\theta=\bigwedge^{m'}_{i=1}p'_i\rightarrow\bot$ and $\Hmsf\cap\Prop(\tau)=\varnothing$. We define
\begin{align*}
\eta^\Cmsf&\coloneqq\bigodot\limits^m_{i=1}(p_i\!\geq\!\mathbf{1})\rightarrow(q\geq\mathbf{1})
&
\theta^\Cmsf&\coloneqq\bigodot\limits^{m'}_{i=1}(p'_i\!\geq\!\mathbf{1})\rightarrow\bot
\end{align*}
and set $\Pmbb^\Cmsf=\langle\Gamma^\Cmsf,\tau^\Cmsf,\Hmsf^\Cmsf\rangle$ as follows:
\begin{align}\label{equ:HorntoCFreduction}
\Gamma^\Cmsf&=\left\{\eta^\Cmsf\mid\eta\in\Gamma\right\}\cup\left\{\theta^\Cmsf\mid\theta\in\Gamma\right\}\nonumber\\
\tau^\Cmsf&=\bigodot_{q\in\tau}(q\!\geq\!\mathbf{1})\odot\bigodot_{\neg q\in\tau}(q\!\leq\!\mathbf{0})\\
\Hmsf&=\{r\!\geq\!\mathbf{1}\mid r\in\Hmsf\}\cup\{s\leq\mathbf{0}\mid\neg s\in\Hmsf\}\nonumber
\end{align}
It is clear that $\Gamma^\Cmsf$ is cover-free and that the size of $\Pmbb^\Cmsf$ is linear in the size of~$\Pmbb$. Now assume that $\sigma\in\Sol(\Pmbb)$. We show that $\sigma^\Cmsf\in\Sol(\Pmbb^\Cmsf)$. First, it is clear that $\Gamma^\Cmsf,\sigma^\Cmsf\not\models_\Luk\bot$ because $\Gamma,\sigma\not\models_\CPL\bot$. To check the entailment, observe that there is a~unit resolution inference of literals comprising $\tau$ from $\Gamma\cup\{\sigma\}$. It is clear that classical unit resolution is sound for interval clauses: if $v(p\geq\mathbf{1})=1$ and $v(((p\!\geq\!\mathbf{1})\odot\tau)\rightarrow\lambda)=1$, then $v(\tau\rightarrow\lambda)=1$; if $v(p\leq\mathbf{0})=1$ and $v(\tau\rightarrow(p\geq\mathbf{1}))=1$, then $v(\tau\rightarrow\bot)=1$. Thus, reusing the classical inference in~$\Luk$, we obtain $\Gamma^\Cmsf,\sigma^\Cmsf\models_\Luk\tau^\Cmsf$, as required. Conversely, let $\sigma\in\Sol(\Pmbb^\Hmc)$ and define $\sigma_\cl=\bigwedge\limits_{(r\!\geq\!\mathbf{1})\in\sigma}\!\!\!\!r\wedge\bigwedge\limits_{(s\!\leq\!\mathbf{0})\in\sigma}\!\!\!\!\neg s$. It is clear that $\Gamma,\sigma_\cl\models_\CPL\tau$. To see that $\Gamma,\sigma_\cl\not\models_\CPL\bot$, assume for contradiction that \mbox{$\Gamma,\sigma_\cl\!\models_\CPL\!\bot$.} Then there is a~unit resolution inference of the empty clause from $\Gamma\cup\{\sigma_\cl\}$. But as we can reuse it in~$\Luk$, this would mean that $\Gamma^\Cmsf,\sigma\models_\Luk\bot$, contrary to the assumption. Contradiction.

As $\Hmsf\cap\Prop(\tau)=\varnothing$, it follows that all solutions to $\Pmbb$ and $\Pmbb^\Cmsf$ are proper. Hence, the reduction works for both arbitrary and proper solutions. Finally, to see that recognition of entailment-minimal solutions is $\Pmsf$-hard, one can easily check that if $\sigma\in\LukminSol(\Pmbb)$, then $\sigma^\Cmsf\in\LukminSol(\Pmbb^\Cmsf)$. Conversely, if $\sigma'\in\LukminSol(\Pmbb^\Cmsf)$, then $\sigma'_\cl\in\LukminSol(\Pmbb)$.
\end{proof}
\coverfreesolutionexistence*
\begin{proof}
The membership follows from Theorem~\ref{theorem:coverfreesolutionrecognition}. For the hardness, we reuse the reduction shown in~\eqref{equ:HorntoCFreduction}.
\end{proof}
\end{document}